\documentclass[%
 reprint,
 preprint,
showpacs,preprintnumbers,
 amsmath,amssymb,
 aps,
 pra,
 longbibliography,
 lengthcheck,%
]{revtex4-1}
\usepackage[usenames,dvipsnames,svgnames,table]{xcolor}
\usepackage{empheq}
\usepackage{enumerate}
\usepackage{type1cm}
\usepackage{graphicx}
\usepackage{epstopdf}
\DeclareGraphicsExtensions{.eps,.png}
\usepackage{amsthm}
\usepackage{epigraph}

\newcommand{\<}[1]{}
\newcommand{\pder}[2]{\mathop{\frac{\partial #1}{\partial #2}}}
\newcommand{\der}[2]{\mathop{\frac{d #1}{d #2}}}
\newcommand{\idx}[1]{\mbox{\textnormal{\scriptsize #1}}}

\newcommand{\arccot}{\mathrm{arccot}}
\newcommand{\Sign}{{\mathrm{sign}}}
\newcommand{\Real}{\mathbb{R}}

\newcommand{\umax}{u_{\idx{max}}}

\newcommand{\Tr}{\mathop{\rm{Tr}}}

\newcommand{\vecu}{\mathbf{u}}

\newcommand{\const}{\mathop{\textrm{const}}}

\newcommand{\ncp}{n}
\newcommand{\rM}{\tilde r^{-}}
\newcommand{\rP}{\tilde r^{+}}

\newtheorem{llemma}{Lemma}

\newtheorem{theorem}{Statement}
\newtheorem{proposition}[theorem]{Proposition}

\let\oldgather = \gather
\let\endoldgather = \endgather
\renewenvironment{gather}[0]{\par\nobreak\noindent\oldgather}{\endoldgather}
\let\oldalign = \align
\let\endoldalign = \endalign
\renewenvironment{align}[0]{\par\nobreak\noindent\oldalign}{\endoldalign}

\begin{document}
\title{Role of control constraints in quantum optimal control }
\author{Dmitry V. Zhdanov}
\email{dm.zhdanov@gmail.com}
\affiliation{Department of Chemistry, Northwestern University, Evanston, IL 60208 USA}
\author{Tamar Seideman}
\email{t-seideman@northwestern.edu}
\affiliation{Department of Chemistry, Northwestern University, Evanston, IL 60208 USA}

\begin{abstract}
\epigraph{Have you stopped drinking brandy in the mornings? Answer: Yes or No?}{Astrid Lindgren\\Karlsson-on-the-Roof is Sneaking Around Again}
The problems of optimizing the value of an arbitrary observable of the two-level system at both a fixed time and the shortest possible time is theoretically explored. Complete identification and classification along with comprehensive analysis of globally optimal control policies and traps (i.e. policies which are locally but not globally optimal) is presented. The central question addressed is whether the control landscape remains trap-free if control constraints of the inequality type are imposed. The answer is astonishingly controversial, namely, although  formally it is always negative, in practice it is positive provided that the control time is fixed and chosen long enough.
\end{abstract}
\maketitle

\section{Introduction\label{@SEC:Intro}}
Coherent control of the two level system is crucial for qubit design. The two-level Landau-Zener system is probably  the most fundamental qubit model with the single control parameter $u$. Its master equation reads:
\begin{gather}\label{01.-evolution_law}
\rho(\tau){=}U_{\tau,0}(u)\rho(0)U_{\tau,0}^{\dagger}(u),
\end{gather}
with the unitary transformation $U_{\tau'',\tau'}(u)$ defined as
$U_{\tau'',\tau'}(u){=}\overrightarrow{\exp}({-}i\int_{\tau{=}\tau'}^{\tau''}(\hat\sigma_x{+}u(\tau)\hat\sigma_z)d\tau)$.
Here $\rho$ is the system's density matrix, $\sigma_{x}$ and $\sigma_{z}$ are Pauli matrices, $\tau$ is a dimensionless time $\tau{=}\alpha t$, and the control parameter is usually proportional to the interaction strength with an external controlled electric or magnetic field ($u{=}\beta{\cal E}$ or $u{=}\beta{\cal B}$).
Depending on the physical meaning of the scaling factors $\alpha$ and $\beta$, Eq.~\eqref{01.-evolution_law} can represent the wide variety of modern experiments on magnetic or optical control of quantum dots \cite{2013-Greve}, vacancy centers in crystals \cite{2012-Kodriano}, spin states of atoms and molecules \cite{2010-Lapert},  Bose-Einstein condensates \cite{2014-Schafer,2013-Malossi}, superconducting circuits \cite{2010-Shevchenko} etc.

We consider the following optimal control problem:
\begin{gather}\label{01.-performance_index}
J{=}\Tr[\rho(T)\hat O]{\to}\max
\end{gather}
\begin{gather}
{-}\umax{\leq}u{\leq}\umax\label{01.-finite_domain},\\
T{<}T_{\idx{max}}\label{01.-finite_time},
\end{gather}
where $\max$ is taken with respect to the program (or control policy) $\tilde u(\tau)$, and possibly also the final time $T$. This task well represents the initial preparation of the qubit in the given initial pure state corresponding to the largest eigenvalue of the observable $\hat O$.

The key question of our study is the extent to which the restrictions \eqref{01.-finite_domain}, \eqref{01.-finite_time} complicate finding the policy $\tilde u^{\idx{opt}}(\tau)$ which maximizes $J[u]$ using the local search methods. The applied value of this question is justified both by technical limitations and also the breakdown of the two-level approximation in strong fields. In addition, the bound \eqref{01.-finite_time} is motivated by the fatal losses of fidelity of quantum gates due to incontrollable decoherence at long times.

The major obstacles in searching for $\tilde u^{\idx{opt}}(\tau)$ are ``traps'' and ``saddle points''. These are the special policies $\tilde u$ such that $J[\tilde u(\tau)]{<}J^{\idx{opt}}$ and $J[\tilde u{+}\delta u]{<}J[\tilde u]$ (trap) or $J[\tilde u{+}\delta u]{=}J[\tilde u]$ (saddle point) for any infinitesimal variation $\delta u(\tau)$ consistent with \eqref{01.-finite_domain}. Presence of saddle points slows down the convergence of any local search strategy whereas reaching the trap leads to its failure. Traps and saddle points are controversial matters for optimal quantum control (OQC) theory. Despite substantial experimental evidence \cite{2011-Moore} supported by theoretical arguments of trap-free ``quantum landscapes'' $J[\tilde u(\tau)]$  \cite{2008-Wu,2009-Brif,2008-Pechen,2006-Ho}, it is also easy to provide simple counterexamples \cite{2012-Hoff,2013-Fouquieres,2011-Pechen}.

The Landau-Zener system is quite special from this perspective since is the only system for which absence of traps in the unconstrained case (i.e{.} when $\umax{=}\infty$ in \eqref{01.-finite_domain}) was formally proven \cite{2001-Khaneja,2012-Pechen}. Moreover, its complete controllability for any finite value of $\umax$ (provided that $T_{\max}$ is chosen sufficiently long) was also justified \cite{1972-Jurdjevic,2000-D'Alessandro,2002-Wu}. Thus, this system provides opportunity to evaluate the effect of constraints \eqref{01.-finite_time} and \eqref{01.-finite_domain} on the landscape complexity in the most pristine form. The existing data portend that this effect should be nontrivial. For example, the unconstrained time-optimal policies $\tilde u(\tau)$ are shown to be $\tilde u(\tau){=}c'\delta(\tau){+}c''\delta(\tau{-}T)$ where $c'$ and $c''$ are constants and $\delta(\tau)$ is the Dirac delta function \cite{2013-Hegerfeldt}. Such solutions are evidently inconsistent with any constraints of form \eqref{01.-finite_domain}.

An additional feature of the Landau-Zener system is its simplicity, which allows us to infer analytically the topology of $J[u]$. We hope that this exceptional opportunity will provide some clue on the expected outcomes in more complex cases, which can be handled only numerically.

It is worth mentioning that the restrictions \eqref{01.-finite_domain} are critical in the foundation of modern theory of optimal control since the corresponding problems can not be solved in the framework of classical calculus of variations and require special methods, such as the Pontryagin's maximum principle (PMP) \cite{BOOK-Pontryagin-1962,BOOK-Agrachev/Sachkov}. For completeness of the presentation, we provide in Sec.~\ref{@SEC:preliminaries} the brief overview of PMP and the known results of the first-order analysis of controlled Landau-Zener system in the PMP framework. In particular, we clarify why the unconstrained problem \eqref{01.-performance_index} is trap-free, and introduce the primary classification of the stationary points (i.e{.} the locally and globally optimal solutions, traps and saddle points) by showing that all of them in the case of time-optimal control as well as traps and saddle points in the case of fixed time control are represented by piecewise-constant controls $\tilde u(\tau)$ which can take only 3 values: 0 and $\pm\umax$.

The rest of the paper is organized as follows. In Sec.~\ref{@SEC:beyond_1st_order} we derive the comprehensive set of criteria which allow to outline the landscape profile and distinguish between various types of its stationary points. The obtained criteria substantially extend, generalize and/or specialize a number of known results \cite{BOOK-Sussman/Agrachev,2005-Boscain,2006-Boscain,2013-Hegerfeldt} obtained for related problems using the index theory \cite{1998-Agrachev} methods, optimal syntheses on 2-D manifolds \cite{BOOK-Boscain}, etc. In this work we propose the technique of ``sliding'' variations which allows to reduce the high-order analysis to methodologically simple and intuitively appealing geometrical arguments.

In subsequent sections~\ref{@SEC:traps(time-optimal_control)} and \ref{@SEC:traps(T=const)} we apply these criteria to identify and classify the traps and saddle points for the cases of time-optimal and time-fixed control, correspondingly. A brief summary of the obtained results and the general conclusions which follow from this analysis are given in the final section~\ref{@SEC:conclusion}.

\section{Regular and singular optimal policies\label{@SEC:preliminaries}}
In this section we review the first-order analysis of problem \eqref{01.-evolution_law} with constraint \eqref{01.-finite_domain} in the PMP framework. For additional details one can refer to extensive literature (e.g.~\cite{BOOK-Agrachev/Sachkov}, pp.280-286,\cite{2005-Boscain}). PMP provides the necessary criterion of local optimality of control $u(\tau)$ in terms of the Hamilton-type Pontryagin function $K(\rho(\tau),\hat O(\tau),u(\tau))$:
\begin{gather}\label{02.-PMP}
\tilde \vecu(\tau){=}\arg\max_{u(\tau)}K(\tilde\rho(\tau),\tilde{\hat{O}}(\tau),u(\tau)).
\end{gather}
The processes satisfying the PMP are called stationary points, or extremals, and will be denoted hereafter with the ~$\tilde{}$~ marks: $\{\tilde u(\tau),\tilde\rho(\tau),\tilde{\hat{O}}(\tau)\}$.

The Pontryagin function among the state variables $\rho$ and controls $u$ depends also on the so called co-state of adjoint variables $\hat O(\tau)$, which are subject to the special evolution equation and boundary transversality conditions. In the case of the control problem \ref{01.-evolution_law}, \eqref{01.-performance_index}, the Pontryagin function takes the form:
\begin{gather}\label{02.-PF}
K(\rho(\tau),\hat O(\tau),u(\tau)){=}{-}i\Tr\left\{[\rho(\tau),\hat O(\tau)](\hat\sigma_x{+}u(\tau)\hat\sigma_z)\right\},
\end{gather}
the evolution equation for $\hat O(\tau)$ coincides with \eqref{01.-evolution_law}:
\begin{gather}
\hat O(\tau''){=}U_{\tau'',\tau'}(u)\hat O(\tau')U^{\dagger}_{\tau'',\tau'}(u)\label{02.-O(t)},
\end{gather}
and the boundary conditions read,
\begin{gather}
\label{02.-O(T)}\hat O(T){=}\hat O;\\
K(T){\begin{cases}{=}0 & \text{ if  $T$ is unconstrained;}\\{\geq}0 & \text{in the case \eqref{01.-finite_time}.}\label{02.-K(T)}\end{cases}}
\end{gather}
Since the Pontryagin function \eqref{02.-PF} linearly depends on $u(\tau)$ the PMP can be satisfied in two ways,
\begin{enumerate}[1)]
\item The switching function $\pder{}{u(\tau)}K{=}{-}i\Tr\left\{[\rho(\tau),\hat O(\tau)]\hat\sigma_z\right\}{\neq}0$. In this case $\tilde u(\tau){=}u_{\idx{max}}\Sign(\pder{}{u(\tau)}K)$, and the corresponding section of the trajectory is called regular. It is clear that the optimal policy $\tilde u(\tau)$ is actively constrained, and relaxing the inequalities \eqref{01.-finite_domain} will improve the optimization result. For this reason, the optimal trajectory containing the regular sections can not be kinematically optimal. An optimal process $\{\tilde\rho(\tau),\tilde{\hat{O}}(\tau),\tilde u(\tau)\}$ for which $\tilde u(\tau){=}\pm\umax$ everywhere except for the finite number of time moments is often referred as the bang-bang control.
\item It may happen that the switching function remains equal to zero over a finite interval of time. In this case, the corresponding segment of the trajectory is called singular, and the associated optimal control can be determined only from higher-order optimality criteria, such as the generalized Legendre-Clebsch conditions, Goh condition etc. \cite{BOOK-Agrachev/Sachkov,1966-Goh,BOOK-Milyutin-1998}
\end{enumerate}

Substituting 
\eqref{01.-evolution_law} and \eqref{02.-O(t)} into \eqref{02.-PF} one can directly check that the Pontryagin function for problem \eqref{01.-evolution_law} is constant along any extremal:
\begin{gather}\label{03.-PF=const<=0}
\forall \tau:K(\tau){=}\tilde K{\geq}0 \mbox{ on each extremal},
\end{gather}
where the strict inequality holds only if the constraint \eqref{01.-finite_time} is active, and
\begin{gather}\label{02.-K(t)=0}
\forall \tau:K(\tau){\equiv}0 \text{ for any kinematically optimal solution}.
\end{gather}

\subsection{Singular extremals of the problem (\texorpdfstring{\ref{01.-evolution_law}}{Lg})}

Every kinematically optimal solution $\tilde u(\tau)$ consist of a single singular subarc. Here we  show that in the case of the Landau-Zener system the converse is also true: every singular extremal $\tilde u(\tau)$ corresponding to inactive constraint \eqref{01.-finite_time} delivers the global kinematic extremum (maximum or minimum) to the problem \eqref{01.-performance_index}. Indeed, let $\tau_1$ be an arbitrary internal point of the singular trajectory. Then, the PMP states that:
\begin{gather}\label{02.-singular_arc_criterion}
\pder{}{u(\tau)}K(\tau){=}{-}i\Tr\left\{[\rho(\tau_1),\hat O(\tau_1)]U^{\dagger}_{\tau,\tau_1}(\tilde u)\hat\sigma_zU_{\tau,\tau_1}(\tilde u)\right\}{\equiv}0
\end{gather}
for any $\tau$ such that $|\tau-\tau_1|{<}\epsilon$ and sufficiently small $\epsilon$, and in particular:

\begin{subequations}\label{02.-singular_arc_sigma_z,y,x}
\begin{gather}\label{02.-singular_arc_sigma_z}
{-}i\Tr\left\{[\rho(\tau_1),\hat O(\tau_1)]\hat\sigma_z\right\}{=}0.
\end{gather}
The two subsequent time derivatives of the equality \eqref{02.-singular_arc_criterion} at $\tau{=}\tau_1$ give,
\begin{gather}
{-}i\Tr\left\{[\rho(\tau_1),\hat O(\tau_1)]\hat\sigma_y\right\}{=}0\\
{-}i\tilde u(\tau_1)\Tr\left\{[\rho(\tau_1),\hat O(\tau_1)]\hat\sigma_x\right\}{=}0.
\end{gather}
\end{subequations}
Equations \eqref{02.-singular_arc_sigma_z,y,x} can be simultaneously satisfied only in the two cases:
\begin{subequations}\label{02.-singular_extremal_criteria}
\begin{align}
&[\rho(\tau),\hat O(\tau)]{=}0; \label{02.-singular_extremal_criterion_1}\\
&[\rho(\tau),\hat{O}(\tau)]{=}i\kappa\hat\sigma_x~~\mbox{and}~~u(\tau){=}0~~(\kappa{=}\const{\ne}0)\label{02.-singular_extremal_criterion_2}.
\end{align}
\end{subequations}
The condition \eqref{02.-singular_extremal_criterion_1} is nothing but the criterion of the global kinematic extremum (maximum or minimum) for our two-level system. In other words, we just proved that all the extrema of the landscape $J(u)$ for the unconstrained Landau-Zener system except for the case of $u(t){\equiv}0$ are its global kinematic maxima and minima. This result was obtained in \cite{2001-Khaneja,2012-Pechen}.

The condition \eqref{02.-singular_extremal_criterion_2} indicates that the only possible everywhere singular non-kinematic extremal of the problem \eqref{01.-performance_index} is $\tilde u(\tau){\equiv}0~ (\tau{\in}[0,T])$. Eq.~\eqref{02.-PF} implies that $K(\tau){=}\kappa$ in this case. Thus, in view of \eqref{02.-K(T)}, this extremal can appear only under the active pressure of the constraint \eqref{01.-finite_time}.

\subsection{Regular and mixed extremals of the problem (\texorpdfstring{\ref{01.-evolution_law}}{Lg})}

According to the PMP and conditions \eqref{02.-singular_extremal_criteria}, the generic non-singular extremal is the piecewise-constant function with $\ncp$ switchings of either bang ($u{=}{\pm}\umax$) or bang-singular ($u{=}{\pm}\umax,0$) type where the singular arcs match \eqref{02.-singular_extremal_criterion_2}. For convenience, we will refer to  extermals with (without) singular arcs as of type II (type I). We will use the subscript $i$ (i.e.~$\tilde\tau_i$, $\tilde\rho_i$ etc., $0{<}i{<}\ncp{+}1$) for the parameters related to the $i$-th control discontinuity (corner point). The durations of the right (left) adjacent arcs and the associated values of $u$ will be labeled as $\tilde{\Delta\tau}_{i}$ ($\tilde{\Delta\tau}_{i{-}1}$) and $\tilde u_i^{+}$ ($\tilde u_i^{-}$). The subscripts $i{=}0$ and $i{=}n{+}1$ will be reserved for the parameters of the trajectory endpoints. We will also sometimes use the notations $^s$I and $^s$II with index $s$ denoting the number of times the control changes the sign.

We first address the properties of type I extremals. The necessary condition of the $i$-th corner point is given by eq.~\eqref{02.-singular_arc_sigma_z}. Combining it with \eqref{03.-PF=const<=0} we get,
\begin{gather}\label{02.-[rho,O]}
{-}i[\tilde\rho(\tilde\tau_i),\tilde{\hat O}(\tilde\tau_i)]{=}c_{i,1}\hat\sigma_x{+}c_{i,2}\hat\sigma_y,~~~c_{i,1},c_{i,2}\in\Real,
\end{gather}
where $c_{i,1}{=}0({>}0)$ when the constraint \eqref{01.-finite_time} is inactive(active) and the case $c_{i,1}{\leq}0$ can result from the optimization with fixed $T$. Consider the adjacent $(i{+}1)$-th bang arc. The PMP criterion \eqref{02.-PMP} for its interior reads,
\begin{gather}\label{03.-PMP_at_1st_regular_segment}
\tilde u(\tau)|_{\tau{>}\tilde\tau_i}{=}\arg\max_u\Tr[U_{\tau,\tilde\tau_i}(c_{i,1}\hat\sigma_x{+}c_{i,2}\hat\sigma_y)U_{\tau,\tilde\tau_i}^{-1}\hat\sigma _z]u,
\end{gather}
which gives $\tilde u^+_i{=}\frac{c_{i,2}}{|c_{i,2}|}\umax$. If the $(i{+}1)$-th arc ends with another corner point $\tilde\tau_{i{+}1}$ then it follows from \eqref{03.-PMP_at_1st_regular_segment} that,
\begin{gather}\label{03.-equation_for_next_junction}
\Tr[U_{\tilde\tau_{i+1},\tilde\tau_i}(c_{i,1}\hat\sigma_x{+}c_{i,2}\hat\sigma_y)U_{\tilde\tau_{i+1},\tilde\tau_i}^{-1}\hat\sigma _z]{=}0.
\end{gather}
Condition~\eqref{03.-equation_for_next_junction} can be reduced to the form,
\begin{gather}
c_{i,2}{\sqrt{\umax^2{+}1}}{=}{-}{c_{i,1}\tilde u_i^{+}\tan(\tilde{\Delta\tau_i}\sqrt{\umax^2{+}1})},
\end{gather}
and resolved relative to $\Delta\tau_{i+1}$. Retaining the physically appropriate solutions consistent with eq.~\eqref{03.-PMP_at_1st_regular_segment} we obtain,
\begin{gather}\label{03.-regular_segment_length}
\tilde{\Delta\tau}_{i+1}{=}
\begin{cases}
 \tilde{\delta\tau}_i,  & c_{i,1}{<}0; \\
 \pi\cos(\alpha){-}\tilde{\delta\tau}_i,  & c_{i,1}{>}0,
\end{cases}
\end{gather}
where $\alpha{=}\arctan(\umax)$ and
\begin{gather}\label{03.-regular_segment_length_2}
\tilde{\delta\tau}_i{=}{\arctan\left(\left|\frac{c_{2,i}}{c_{1,i}\umax}\right|\sec(\alpha)\right)}\cos(\alpha).
\end{gather}
Note that ${-}i[\tilde\rho(\tilde\tau_{i+1}),\tilde{\hat O}(\tilde\tau_{i+1})]{=}c_{1,i}\hat\sigma_x{-}c_{i,2}\hat\sigma_y$, i.e.
\begin{gather}\label{03.-c1,c2-relations}
c_{1,i{+}1}{=}c_{1,i},~~c_{2,i{+}1}{=}{-}c_{2,i}.
\end{gather}
Since eqs.~\eqref{03.-regular_segment_length} and \eqref{03.-regular_segment_length_2} do not depend on the sign of $c_{i,2}$ one obtains that durations of all interior bang segments are equal: ${\forall} i{\geq2},i{\leq}\ncp: \tilde{\Delta\tau}_i{=}\tilde{\Delta\tau}$ (see Fig.~\ref{@FIG.01}a). Moreover, eq.~\eqref{03.-regular_segment_length} admits the estimate $\frac{\pi}{2}\cos\alpha{\leq}\Delta\tau{\leq}\pi\cos\alpha$ for the case of time-optimal problem with constraint \eqref{01.-finite_time}.

\begin{figure}[t!]
\ifpdf
  \includegraphics[width=0.45\textwidth]{{fig.01}.eps}
\else
  \includegraphics[width=0.45\textwidth]{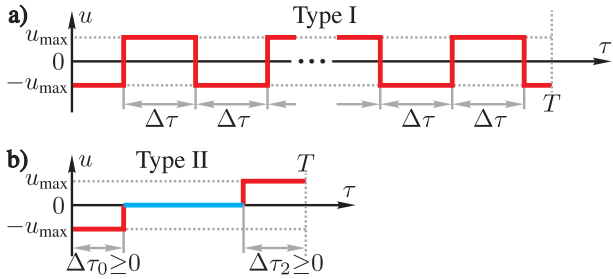} 
\fi
\caption{Possible types of extremals $\tilde u(t)$ associated with nonkinematic optimal solutions and traps as well as the locally time-optimal kinematic optimal solutions.
\label{@FIG.01} }
\end{figure}

Consider now the extremals of type II. Let $\tau{\in}(\tilde\tau_{j{-}1},\tilde\tau_{j})$ be the singular arc where the relations \eqref{02.-singular_extremal_criterion_2} hold. If $\tilde\tau_{j}{\ne}\tilde T$ when it is the corner point between regular and singular arc. Suppose that there exists one more corner point $\tau_{j+1}{>}\tau_{j}$. Then it follows from eqs.~\eqref{03.-c1,c2-relations}~\eqref{03.-regular_segment_length} and \eqref{03.-equation_for_next_junction} that $\tilde{\Delta\tau}_j{=}\pi\cos\alpha$ and $U_{\tilde\tau_{j{+}1},\tilde\tau_j}{=}{-}\hat I$, so that $\tilde\rho(\tilde\tau_{j+1}){=}\tilde\rho(\tilde\tau_{j})$. Using  similar arguments, it is straightforward to derive the analogous result for possible corner points prior to $\tau_j$. Thus, taking any 3-segment ``anzatz'' extremal similar to that shown in Fig.~\ref{@FIG.01}b, one can construct an infinite family ${\cal F}^{[\mathrm k]}(\tilde u(\tau))$ of II$^{[\mathrm k]}$ extremals ($\mathrm k{=}{k_1,k_2}$) by randomly inserting $k_1$ and $k_2$ bang segments of the length $\pi\cos\alpha$ with $u{=}{+}\umax$ and  $u{=}{-}\umax$ into corner points of $\tilde u(\tau)$ or inside its singular arcs. It is clear that each family ${\cal F}^{[\mathrm k]}(\tilde u(\tau))$ constitutes the connected set of solutions, and all the members have equal performances $J$. Thus, the properties of any type II extremal can be reduced to the analysis of the equivalent three-segment $^0$II type or $^1$II type extremal where all the positive and negative bang segments are merged into distinct continuous arcs separated by a singular arc.

The presented first-order analysis outlines the admissible profiles for optimal non-kinematic solutions (see Fig.~\ref{@FIG.01}). Moreover, by continuity argument (i.e. by considering the series of solutions with fixed $T{\to}T_{\mathrm{opt}}$ from below), these profiles should embrace all possible types of the stationary points of the time-optimal problem \eqref{01.-performance_index},\eqref{01.-finite_time}. It is worth stressing that the later include the globally optimal and everywhere singular kinematic solutions for which both segments with $u{=}\pm\umax$ and $u{=}0$ are singular. With this in mind, it is helpful to introduce the following terminological convention for the rest of the paper in order to preserve the integrity of the presentation while avoiding potential confusions: we will reserve the term ``singular'' exclusively for the segments of extremals at which $u{=}0$ whereas the segments with $u{=}\pm\umax$ will be always referred to as ``bang'' ones.

The reviewed results have several serious limitations. First, they do not allow to distinguish the globally time-optimal solution from the trap or saddle point. Second, they do not provide detailed \emph{a priori} knowledge of the characteristic structural features of these stationary points (e.g{.} the expected type, number of switchings etc.) which is necessary to determine the topology of the landscape $J[u]$. These tasks require higher-order analysis, which is the subject of the next section.

\section{Detailed characterization of the stationary points \label{@SEC:beyond_1st_order}}
In this section we will extensively use the geometrical arguments in our reasoning. To make the presentation more visual, it is useful to expand the states and observables in the basis of Pauli matrices and identity matrix $\hat I$: $\rho{=}\frac12\hat I{+}\sum_{i{=}x,y,z}r_i\hat\sigma_i$, $\hat O{=}\frac12\Tr[\hat O]\hat I{+}\sum_{i{=}x,y,z}o_i\hat\sigma_i$. The dynamics induced by eq.~\eqref{01.-evolution_law} corresponds to the rotation of the 3-dimensional Bloch vector $\vec r=\{r_x,r_y,r_z\}$ around the axis $\vec n_u{\propto}\{1,0,u\}$ (note that the angle between $\vec n_{{\pm}\umax}$ and $\vec n_{0}$ is equal to $\alpha$, see e.g{.} Fig.~\ref{@FIG.02}), and the optimization goal \eqref{01.-performance_index} is equivalent to the requirement to arrange the state vector $\vec r$ in parallel to $\vec o$. In what follows we will often refer to the quantum states $\rho$ as the endpoints $r$ of vectors $\vec r$. Hereafter we will also assume that both $r$ and $o$ are renormalized (scaled) such that $|r|{=}|o|{=}1$.

We  start by taking a closer look at type II extremals and their singular arc(s) where $\tilde u(\tau){=}0$. According to criterion \eqref{02.-singular_extremal_criterion_2}, these arcs are always located in the equatorial plane $x{=}0$. The following proposition indicates that such arcs may represent the time-optimal solution at any values of $\umax$ (see Appendix~\ref{@APP:1} for proof):
\begin{proposition}\label{*proposition:equatorial_always_type_II}
The shortest type II singular trajectory connecting any two ``equatorial'' points $\vec r^-{=}\{0,r^-_y,r^-_z\}$ and $\vec r^+{=}\{0,r^+_y,r^+_z\}$ (see Fig.~\ref{@FIG.02}) will represent the (globally) time-optimal solution if $r^-_y r^+_y{>}0$, $(r^+_z{-}r^-_z)r^-_y{>}0$ and the saddle point otherwise.
\end{proposition}

\begin{figure}[t!]
\textbf{\begin{tabular}{cc}
\begin{minipage}{0.5\columnwidth}\includegraphics[width=1\columnwidth]{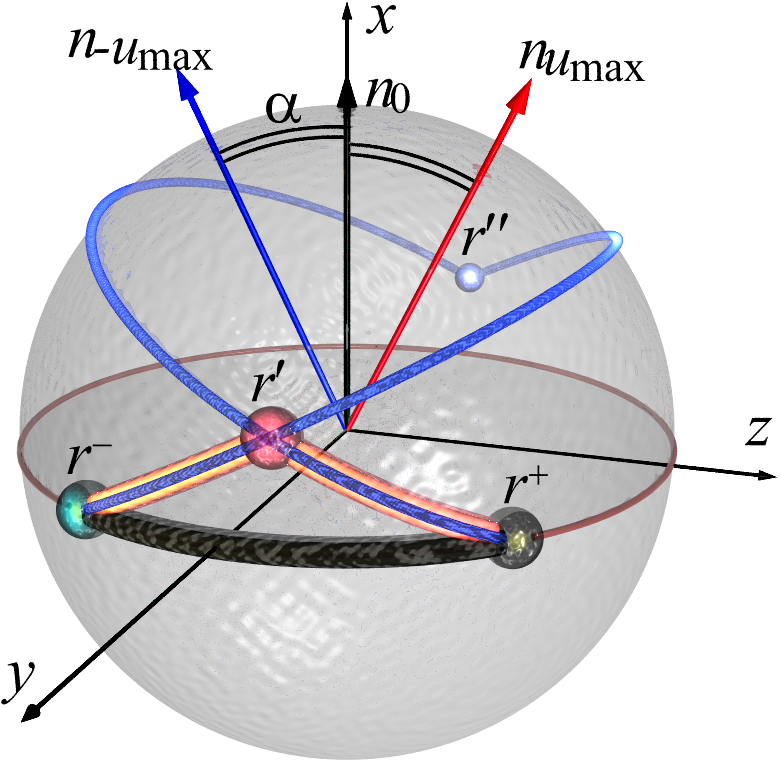}\\
a)\end{minipage} & \begin{minipage}{0.5\columnwidth}\includegraphics[width=1\columnwidth]{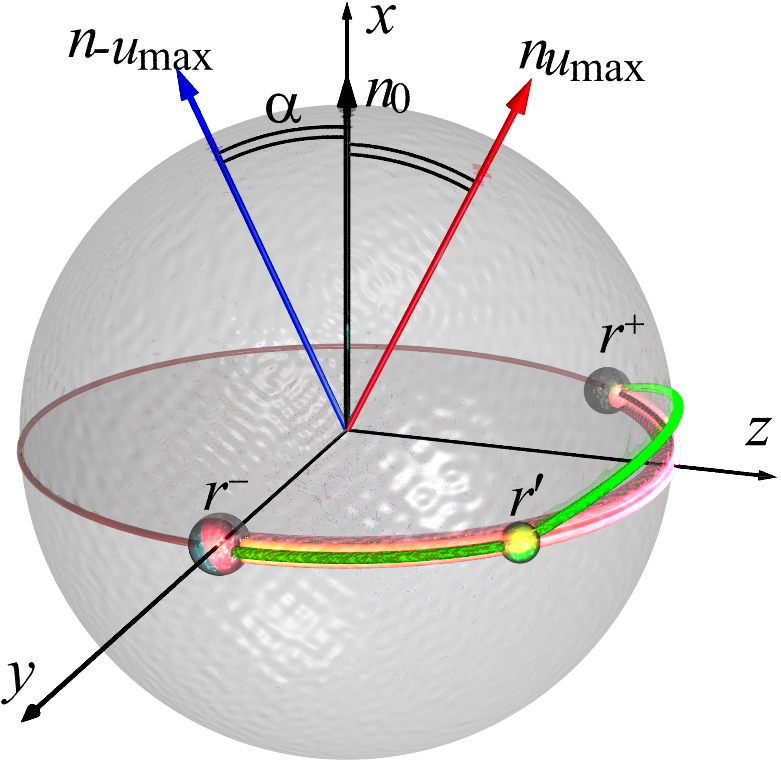}\\
b)\end{minipage}
\end{tabular}}
\caption{(a) The case $r^-_yr^+_y{>}0$: Equatorial singular arc $r^-{\to}r^+$ (thick black line) is more time-effective than bang-bang extremal $r^-{\to}r'{\to}r^+$ (thick orange line). The extremal $r^-{\to}r''{\to}r^+$ (thin blue curve) represents local extremum (trap). (b) The case $r^-_yr^+_y{<}0$: Equatorial singular arc $r^-{\to}r^+$ (thick orange line) is the suboptimal relative to the bang-singular extremal $r^-{\to}r'{\to}r^+$ (thin black line). \label{@FIG.02} }
\end{figure}

Since all $^s$II extremals can be reduced to the effective 3-section anzatz (see the end of the previous section) Proposition~\ref{*proposition:equatorial_always_type_II} has the evident corollary:
\begin{proposition}\label{*proposition:singular_sections_on_one_xz_side}
All singular arcs of the locally optimal type II extremals are located in the same semi-space $y{>}0$ or $y{<}0$, and their total duration can not exceed $\pi/2$.
\end{proposition}

For further analysis we need the following generic necessary condition of the time optimality:
\begin{proposition}\label{*proposition:s*r_x*r_y<0}
If the type I extremal $\{\tilde u(\tau),\tilde r(\tau)\}$ is locally time-optimal then each of its corner points $\tilde r_i$ satisfies the inequality:
\begin{gather}\label{04.-proposition:s*r_x*r_y<0}
\tilde u_i^- \tilde r_{i,x} \tilde r_{i,y}{\geq}0.
\end{gather}
\end{proposition}

Qualitatively, Proposition~\ref{*proposition:s*r_x*r_y<0} states that the projections of optimal trajectories on the $xz$-plane are always "V"-shaped at the corner points $\tilde r_i$ with $\tilde r_{i,x}{>}0$ and "$\Lambda$"-shaped otherwise (here we assume that the $x$-axis is oriented vertically, like in Fig.~\ref{@FIG.02}).

Proposition~\eqref{*proposition:s*r_x*r_y<0} allows to substantially narrow down the range of the type II candidate trajectories:
\begin{proposition}\label{*proposition:type_II_no_interior_bangs}
Any type $^s$II extremal with $s{>}0$ containing the interior bang arc of duration $\Delta\tau{>}\pi\sec\alpha$ is the saddle point for the time-optimal control.
\end{proposition}
In other words, all the type $^s$II$|_{s{>}0}$ locally time-optimal solutions reduce to the 3-piece anzatz shown in Fig.~\ref{@FIG.01}b where two regular arcs of duration $\tilde\Delta\tau_0,\tilde\Delta\tau_2{<}\pi\sec\alpha$
``wrap'' the singular section where $u{=}0$. Accordingly, the  number of control switchings is bounded by $\ncp_{\mathrm{II}}{\leq}2$.

The properties of the $^0$II type extremals are richer:
\begin{proposition}\label{*proposition-0^II_trajectories_are_locally_time-optimal}
Suppose that the $^0$II type extremal $\tilde u(\tau)$ is the member of family ${\cal F}^{[\mathrm k]}(\tilde u^{\idx{anz}}(\tau))$, and its anzatz $\tilde u^{\idx{anz}}(\tau)$ has nonzero durations $\tilde\Delta\tau_0$ and $\tilde\Delta\tau_2$ of opening and closing bang segments. Then $\tilde u(\tau)$ is locally optimal iif $\tilde u^{\idx{anz}}(\tau)$ is locally optimal
\end{proposition}
\noindent(for proof see Appendix~\ref{@APP:9}).

\begin{figure}[t!]
   \includegraphics[width=0.35\textwidth]{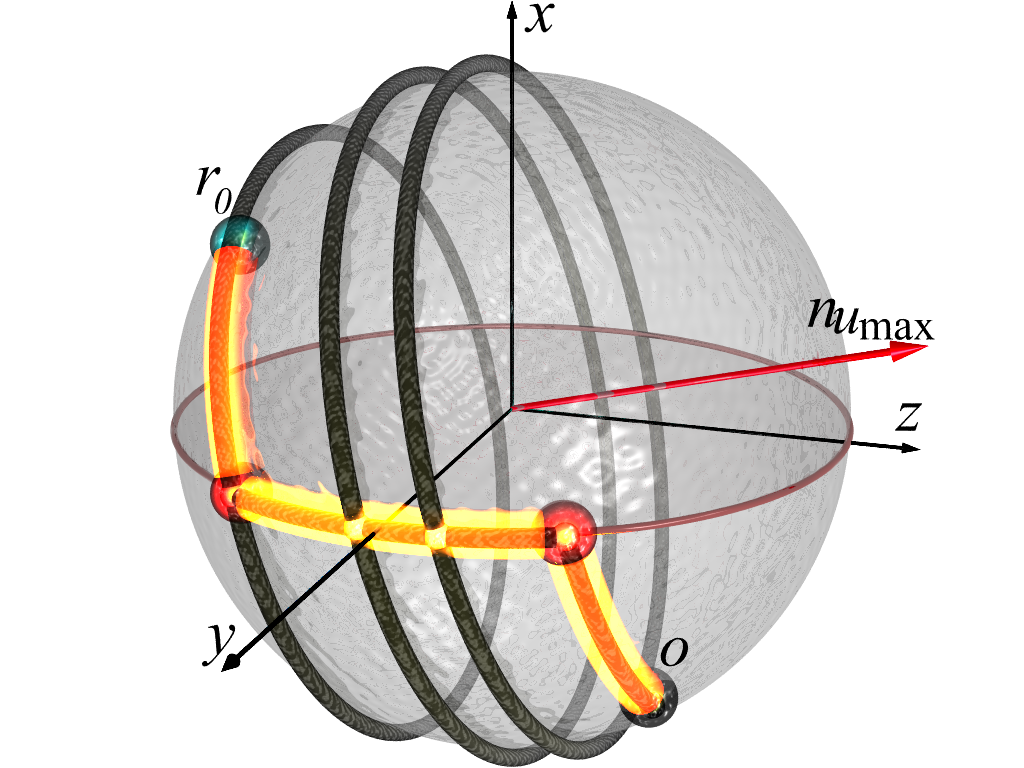}
\caption{The globally time-optimal $^0$II type trajectory $\tilde u_{\idx{anz}}(\tau)$ (thick bright yellow curve) and the locally time-optimal trapping solution (black curve) of the ${\cal F}^{[3]}(\tilde u^{\idx{anz}}(\tau))$ family connecting the points $r_0{\propto}\{1,1,{-}1\}$ and $o{\propto}\{{-}1,1,1\}$.\label{@FIG.11}}
\end{figure}

The analysis of type I extremals is somewhat more complicated. We begin by determining the loci of corner points $\tilde r_i$ on the Bloch sphere. Denote as $\theta{=}2\tilde{\Delta\tau}\sec\alpha$ the rotation angles on the Bloch sphere associated with the inner bang sections of the type I extremals. Note that it follows from \eqref{03.-regular_segment_length}, \eqref{03.-regular_segment_length_2} that $\pi{<}\theta{<}2\pi$ in the case of time-optimal control problem.

\begin{proposition}\label{*proposition:corner_points_locations}
All the corner points $\tilde r_i$ of any locally optimal type I solution $\tilde u(\tau)$ of problem \eqref{01.-performance_index}, \eqref{01.-finite_domain} are located on the circular intersections of the Bloch sphere with the two planes~$\lambda_{\pm1}$ (see Fig.~\ref{@FIG.03}):
\textnormal{\begin{gather}\label{04.-switching_points_r_i-explicit}
\tilde r_i{=}\{\Sign(\tilde u_{i}^+)\sin (\gamma_i)\sin(\frac{\xi }{2}),{-}\sin(\gamma_i)\cos (\frac{\xi}{2}),\cos(\gamma_i)\}
\end{gather}}
\noindent Here \textnormal{$\xi{=}{-}2\arctan\left(\frac{\umax}{2} \tan(\frac{\theta}{2})\cos(\alpha)\right)$}
is the dihedral angle between the planes $\lambda_{\pm1}$, and $\gamma_{i+1}{=}\gamma_1{+}i\eta$,
where \textnormal{$\eta{=}{-}2 \arctan\left(\frac{\sin(\frac{\theta}{2})}{\sqrt{\umax^2{+}\cos ^2(\frac{\theta }{2})}}\right)$}.
\end{proposition}

\begin{proposition}\label{*proposition:second-order_Q-criterion}
Denote $q_i{=}q(\gamma_i){=}\cot^2(\gamma_i){-}\cot^2(\frac{\eta }{2})$ $(i{=}1,...,n)$. The set $\{q_i\}$ associated with any locally time-optimal extremal $\tilde u(t)$ contains at most one negative entry $q'$, and $|q'|{=}\min(|\{q_i\}|)$.
\end{proposition}
The proofs of the above two propositions are given in Appendix~\ref{@APP:4}.
\begin{figure}[t!]
   \ifpdf
      \includegraphics[width=0.45\textwidth]{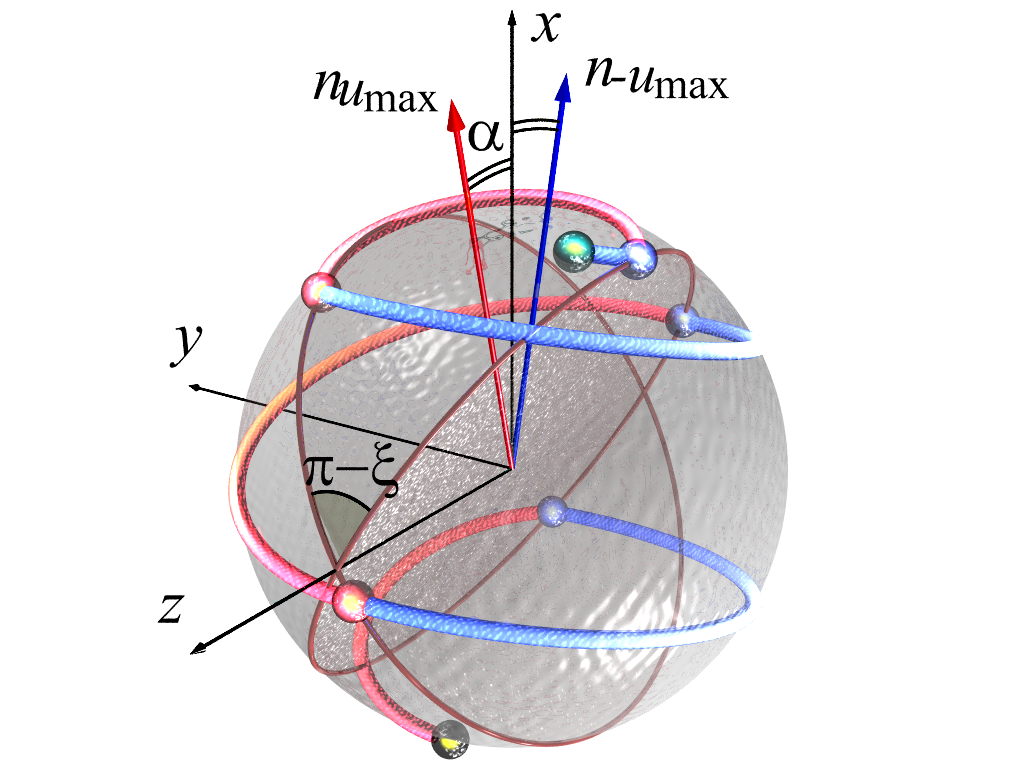}
   \else
   \fi
\caption{Illustration of the statement of Proposition~\ref{*proposition:corner_points_locations}. The thick colored curve depicts the band-bang extremal. Its red and blue segments correspond to $u{=}\max$ and $u{=}{-}\max$. All interior corner points (red and blue balls) lie on two circles (associated with switchings $\umax{\to}{-}\umax$ and ${-}\umax{\to}\umax$, correspondingly) whode planes $\lambda_{\pm1}$ intersect along the $z$-axis.
\label{@FIG.03} }
\end{figure}

To use Proposition~\ref{*proposition:second-order_Q-criterion} it is convenient to introduce the parameters $\zeta_i$ through, $\zeta_1{=}\gamma_1{+}\frac{\pi}{2}(1{-}\Sign(u_1^+))$, $\zeta_{i+1}=\zeta_1{+}i(\pi{+}\eta)$. It is evident that $q(\gamma_i){=}q(\zeta_i)$. The relation between the sign of $q_i$ and the index $i$ of the corner point can be illustrated by associating each $q_i$ with the point on the unit cycle whose position is specified by $\zeta_i$, as shown in Fig.~\ref{@FIG.05}. One can see that the maximal number $n_{\mathrm{max}}$ of sequential parameters $q_i$ having at most one negative term can not exceed $\frac{\pi{+}|\eta|}{\pi{-}|\eta|}{+}1{\leq}\frac{\pi}{\alpha}$, i.e.,

\begin{proposition}\label{*proposition:n<=pi/alpha}
Type I locally optimal extremals can have at most $\frac{\pi}{\alpha}$ switchings.
\end{proposition}
\noindent This helpful upper bound was first obtained by Agrachev and Gamkrelidze \cite{BOOK-Sussman/Agrachev}. As shown in Appendix~\ref{@APP:5}, we can further refine this result via more detailed inspection of the criterion $|q'|{=}\min(|\{q_i\}|)$:
\begin{proposition}\label{*proposition:n<=2_for_u>sqrt(1+sqrt(2))}
\textnormal{
$\ncp_{\mathrm{I,max}}{\leq}2$} if \textnormal{$\umax{>}\sqrt{1{+}\sqrt{2}}$
}
\end{proposition}
\noindent (The latter roughly corresponds to $\alpha{>}1$).

\begin{figure}[t]
\includegraphics[width=0.5\linewidth]{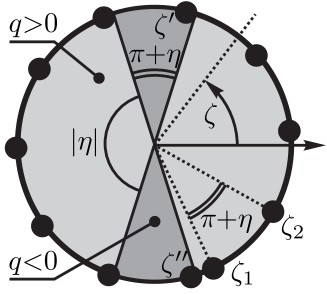}
\caption{Signs of parameters $q(\zeta)$ as function of $\zeta$. Black dots indicate the values $\zeta{=}\zeta_i$ associated with $i$-th corner point.\label{@FIG.05}}
\end{figure}

The analysis in this section so far is equally valid for both global and local extrema (traps) of optimal control. It is clear that any globally time-optimal type II solution includes at most $2$ corner points that separate the central singular section from the outside regular arcs (see Fig.~\ref{@FIG.01}b). The case of type I solutions is not as evident. The following propositions impose more stringent necessary criteria on the globally time-optimal extremals (see Appendices~\ref{@APP:6} and \ref{@APP:7} for proofs).

\begin{proposition}\label{*proposition-n_max_for_time-optimal_case}
Any corner point $\tilde r_{i'}$ such that $q(\gamma_{i'}){<}0$ must be either the first or the last corner point of the globally time optimal solution, so that the total number of switchings $n_{\mathrm{I},\max}{\leq}\frac{\pi}{2\alpha}{+}1$.
\end{proposition}

\begin{proposition}\label{*proposition-r_i-global_optimality_bounds}
The corner points $\tilde r_i$ of any globally optimal solution of type I satisfy the inequality:
\textnormal{\begin{gather}\label{*proposition-r_i-global_optimality_bounds_inequality}
\min(0,\tilde r_{0,x},\tilde r_{n{+}1,x}){<}\tilde r_{i,x}{<}\max(0,\tilde r_{0,x},\tilde r_{n{+}1,x}),
\end{gather}}
where $\tilde r_{0,x}$ and $\tilde r_{n{+}1,x}$ are the trajectory endpoints.
\end{proposition}

Proposition~\ref{*proposition-r_i-global_optimality_bounds} can be used to establish the following more accurate upper bound on the number of switchings (see Appendix~\ref{@APP:8} for proof).
\begin{proposition}\label{*proposition-n_max_for_time-optimal_case_refined}
The number of corner points of the globally time-optimal type I solution $\tilde u(\tau)$ is bounded by the following inequalities:
\begin{widetext}\textnormal{
\begin{subequations}\label{*proposition-n_max_for_time-optimal_case_refined-eqs}
\begin{alignat}{2}[left ={\ncp_{\mathrm{I}}{\leq}\empheqlbrace}]
&\max(
\frac{\arccos(\frac{\rM_x}{\rP_x})}{|2\arctan(\frac{\umax}{\rP_x})|},
\frac{\pi}{|2\arctan(\frac{\umax}{\rM_x})|}
){+}1 &\mbox{ if } \rM_x\rP_x{<}0; \label{*proposition-n_max_for_time-optimal_case_refined(mp<0)}
\\
&\min(
\frac{\arccos(\frac{\rM_x}{\rP_x})}{|2\arctan(\frac{\umax}{\rP_x})|}{+}3,
\frac{\pi}{|4\arctan(\frac{\umax}{\rP_x})|}
){+}1 &\mbox{ if } \rM_x\rP_x{>}0,\label{*proposition-n_max_for_time-optimal_case_refined(mp>0)}
\end{alignat}
\end{subequations}}
\end{widetext}
where $\rP$ and $\rM$ are new notations for the trajectory endpoints $\tilde r_0$ and $\tilde r_{n{+}1}$, such that $|\rP_x|{\geq}|\rM_x|$.
\end{proposition}

Let us denote $\phi_{\xi}{=}\left|\theta_{r_0,\xi}{-}\theta_{o,\xi}\right|$ ($\xi=x,z$), where $\theta_{r,\xi}$ is the angle between the axes $\vec\epsilon_{\xi}$ and $\vec r$. One can geometrically show that the maximal possible change $\Delta\theta^{\max{}}_{r,\xi}$ in $\theta_{r,\xi}$ generated by rotation around any of the axes $\vec n_{\pm\umax}$ is $\Delta\theta^{\max{}}_{r,x}{=}2\alpha$ and $\Delta\theta^{\max{}}_{r,z}{=}\pi{-}2\alpha$ (see Fig.~\ref{@FIG.07}). This fact allows us to establish the following lower bounds on the number of corner points:

\begin{figure}[t!]
   \includegraphics[width=0.45\columnwidth]{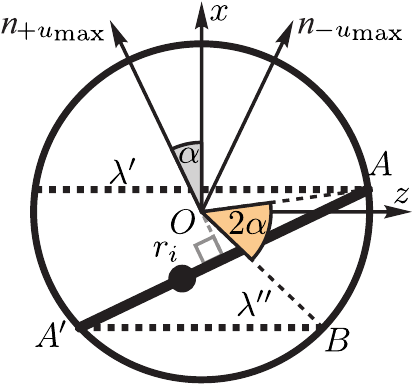} 
\caption{Geometrical calculation of the value of $\Delta\theta^{\max{}}_{r,x}$. Rotation  ${\cal S}_{\vec n_{{-}\umax}}$ around vector $\vec n_{{-}\umax}$ transfers any point $r_i$ on Bloch sphere into new point in $AA'$ plane. The $x$-coordinate of this new point is bounded by planes $\lambda'$ and $\lambda''$. Thus, the associated change in $\theta_{r,x}$ is less than $\angle AOB{=}2\alpha$. \label{@FIG.07}}
\end{figure}

\begin{proposition}\label{*proposition-n_min-bounds}
The minimal number of corner points in locally time-optimal solutions reaching the global maximum of $J$ is bounded by the inequalities:
\textnormal{
\begin{subequations}\label{*proposition-n_min-bounds-eqs}
\begin{gather}
\ncp{\geq}\frac{|\arcsin(r_{0,x}){-}\arcsin(o_x)|}{2\arctan(\umax)}{-}1;\label{*proposition-n_min-bounds-x}\\
\ncp_{\mathrm{I}}{\geq}\frac{|\arcsin(r_{0,z}){-}\arcsin(o_z)|}{2\arccot(\umax)}{-}1\label{*proposition-n_min-bounds-z};
\end{gather}
\end{subequations}
}
\end{proposition}
\noindent It is worth stressing that the bound \eqref{*proposition-n_min-bounds-z} is valid only for type I solutions.

Combination of the upper bounds on $\ncp$ imposed by Propositions~\ref{*proposition:type_II_no_interior_bangs} and \ref{*proposition-n_max_for_time-optimal_case} with inequalities \eqref{*proposition-n_min-bounds-eqs} leads to the following conclusion:
\begin{proposition}\label{*proposition-psi-criteria_of_type_I_II_solutions}
The globally time-optimal solution(s) of problem \eqref{01.-performance_index} is of type I if
\begin{subequations}\label{*proposition-psi-criteria_of_type_I_II_solutions-eqs}
\textnormal{
\begin{gather}\label{*proposition-psi-criteria_of_type_I_II_solutions(x)}
\phi_x{=}\left|{\arcsin(r_{0,x}){-}\arcsin(o_x)}\right|{>}
4\alpha 
\end{gather}
}
\noindent and of type II if
\textnormal{\begin{gather}\label{*proposition-psi-criteria_of_type_I_II_solutions(z)}
\phi_z{=}\left|{\arcsin(r_{0,z}){-}\arcsin(o_z)}\right|{>}\left[\frac{\pi}{2\alpha }{+}2\right](\pi{-}2\alpha).
\end{gather}
}
\end{subequations}
\end{proposition}
Note that this estimate can be further detailed if combined with the refined upper bounds stated in Proposition~\ref{*proposition-n_max_for_time-optimal_case_refined}.
The statement of Proposition~\ref{*proposition-psi-criteria_of_type_I_II_solutions} is graphically illustrated in Fig.~\ref{@FIG.08} which clearly shows that the type I and type II solutions are dominant in the opposite limits of tight and loose control restriction $\umax{\to}0$ and $\umax{\to}\infty$, correspondingly. Neither type, however, totally suppresses the other one at any finite positive value of $\umax$. This coexistence sets the origin for the generic structure of suboptimal solutions (traps), whose analysis will be the subject of next two sections.

\begin{figure}[t!]
   \includegraphics[width=0.25\textwidth]{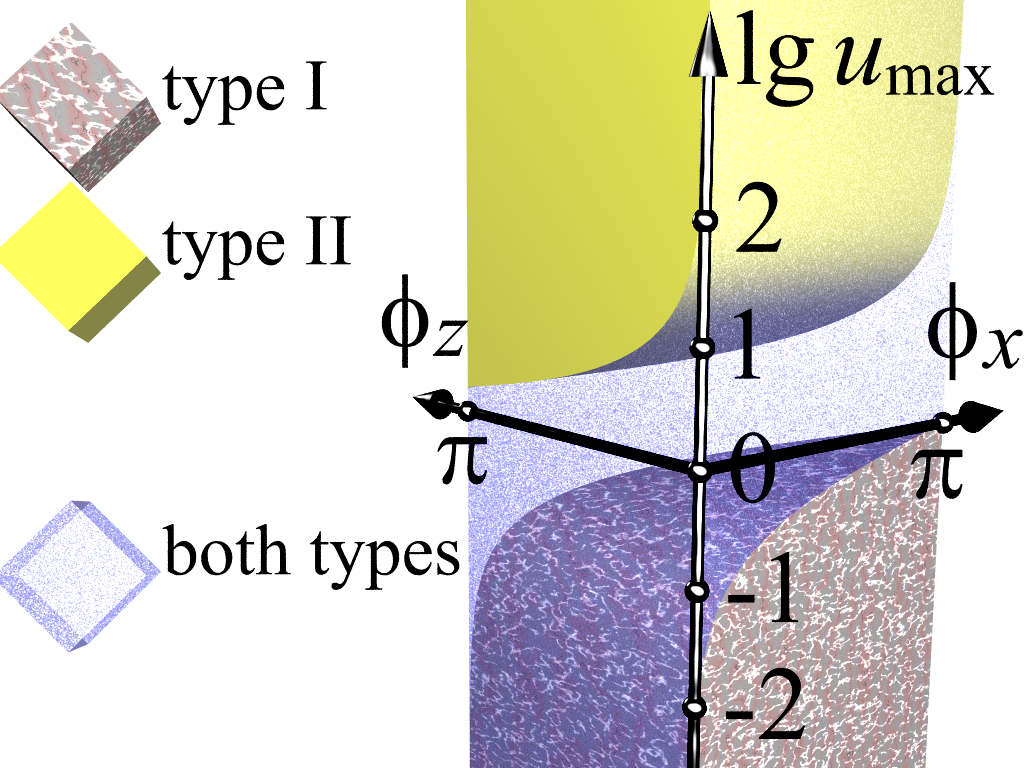}
\caption{Distribution of types of globally optimal solutions according to Proposition~\ref{*proposition-psi-criteria_of_type_I_II_solutions}. Note that the admissible values of $\phi_x$ and $\phi_z$ are restricted by inequality $\phi_x{+}\phi_z{\leq}\pi$.\label{@FIG.08}}
\end{figure}

\section{Traps in time-optimal control \label{@SEC:traps(time-optimal_control)}}

The globally time-optimal solution (hereafter denoted as $\tilde u^{\mathrm{opt}}$) of the problem \eqref{01.-performance_index} can be supplemented by a number of trapping suboptimal solutions $\tilde u$ (characterized by $\tilde J{<}\tilde J^{\mathrm{opt}}$ and/or $\tilde T{>}\tilde T^{\mathrm{opt}}$) which are however optimal with respect to any infinitesimal variation of $\tilde u(\tau)$ and $T$. In particular, Proposition~\eqref{*proposition:singular_sections_on_one_xz_side} implies that each locally optimal solution of type $^0$II gives rise to the infinite family of traps of the form shown in Fig.~\ref{@FIG.11}. In what follows, we will call such traps as "perfect loops". Proposition~\ref{*proposition:equatorial_always_type_II} indicates that the perfect loops may exist at any value of $\umax$. Nevertheless, their presence does not stipulate sufficient additional complications in finding the globally optimal solution by gradient search methods. Indeed, these ``simple'' traps can be identified at no cost by the presence of the continuous bang arc of the duration $\tilde{\Delta\tau_i}{\geq}\pi\sec(\alpha)$. Moreover, one can easily escape any such trap by inverting the sign of the control $u(\tau)$ at any continuous subsegment of this arc of duration $\pi\sec(\alpha)$ or by removing the respective time interval from the control policy.

For this reason, the primary objective of this section is to investigate the other, ``less simple'' types of traps which can be represented by type I and $^s$II$|_{s{>}0}$ suboptimal extremals. Propositions~\ref{*proposition:n<=pi/alpha}, \ref{*proposition-n_max_for_time-optimal_case}, \ref{*proposition-n_max_for_time-optimal_case_refined},  and \ref{*proposition-n_min-bounds} show that the number of switchings $n$ in such extremals is always bounded (at least by $\pi/\alpha$). Thus, the maximal number of such traps is also finite and decreases with increasing $\umax$. It will be convenient to loosely classify the traps into the ``deadlock'', ``loop'' and ``topological'' ones as follows. The first two kinds of traps are represented by type I extremals. The deadlock traps are defined by inequalities $\tilde J{<}\tilde J^{\mathrm{opt}}$ $\tilde T{<}\tilde T^{\mathrm{opt}}$. They usually also satisfy the inequalities $n{<}n^{\mathrm{opt}}$. Their existence is mainly related to the fact that the distance to the destination point $o$ for most of extremals non-monotonically changes with time. The trajectory of the loop trap has the intersection with itself other than the perfect loop. These solutions require longer times $\tilde T{>}\tilde T^{\mathrm{opt}}$ and typically also larger numbers of switchings $n{>}n^{\mathrm{opt}}$ in order to reach the kinematic extremum $\tilde J{=}\tilde J^{\mathrm{opt}}$. Finally, the topological traps are associated with extremals of the type distinct from the type of the globally optimal solution. Of course, real traps can combine the features of all these three kinds.

\newcommand{\tblForFigNine}{
\setlength{\tabcolsep}{0pt}
\begin{minipage}{\columnwidth}
\begin{tabular}
{|c|c|c|c|c|c|}
\hline \rowcolor{gray!30}
~extremal ~&~$\Sign(\tilde u^-_{1})$ ~&~$\ncp$ ~&~$\tilde{\Delta\tau}_1$ ~&~$\tilde{\Delta\tau}$ ~&~$\tilde{\Delta\tau}_{n{+}1}$ ~\\
\hline red      &$+$    &    0 & 0.23 & - & -\\
\hline green    &$-$    &    2 & 0.88 & 1.52 & 0.88\\
\hline blue     &$+$    &    4 & 0.33 & 1.78 & 0.33\\
\hline black    &$-$    &    5 & 1.15 & 1.72 & 0.57\\
\hline
\end{tabular}
\end{minipage}
}

Examples of the deadlock and loop traps are shown in Fig~\ref{@FIG.09}. In this case the globally time optimal solution with $\ncp^{\mathrm{opt}}{=}4$ is accompanied by two deadlock traps and two degenerate loop traps corresponding to $\ncp{=}5$ (only one is shown; the remaining solution can be obtained via subsequent reflections of the black trajectory relative to the $yz$ and $xy$-planes). At the same time, no traps exist for $\ncp{=}1,3$ and $\ncp{>}5$.
\begin{figure}[t!]
   \includegraphics[width=0.35\textwidth]{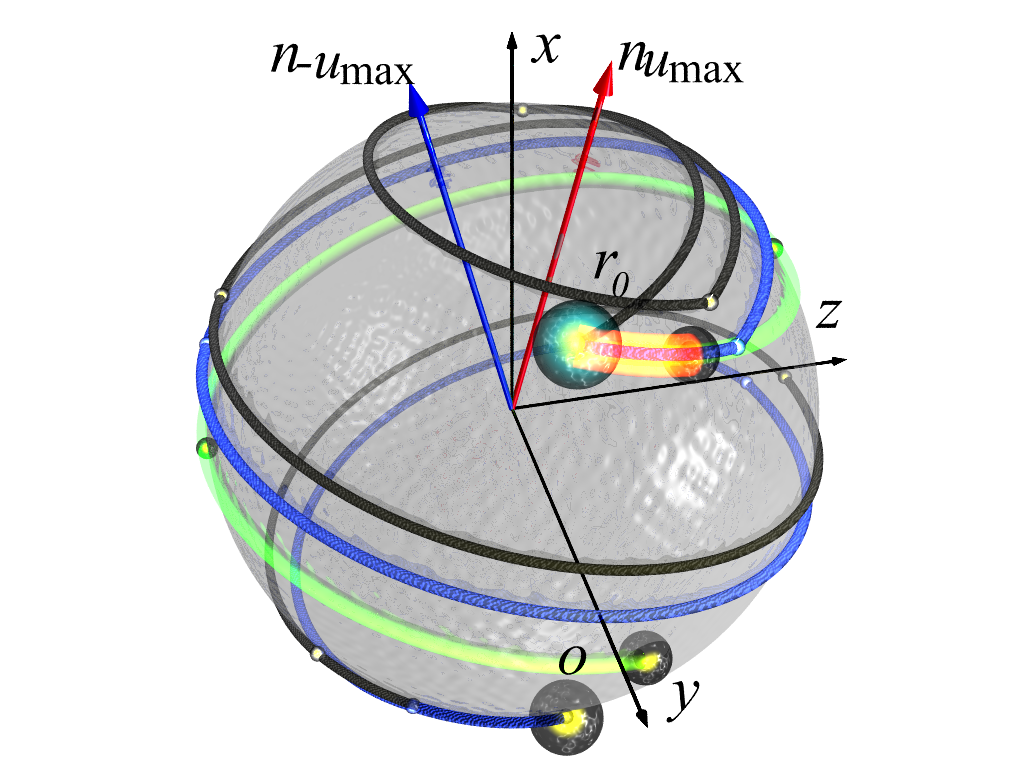}
\caption{Globally optimal solution (blue line), deadlock traps (light-red and green lines) and loop trap (black line) for the time-optimal control problem \eqref{01.-performance_index},\eqref{01.-finite_domain},\eqref{01.-finite_time} with $\umax{=}\frac14$, $r(0){=}\{\frac1{\sqrt 2},\frac1{\sqrt 2},0\},$ (big emerald dot) and $o{=}\{\frac1{\sqrt 2},{-}\frac1{\sqrt 2},0\}$ (big black-yellow dot). Small dots indicate the positions of corner points. The parameters of extremals are listed in the table:\\
\protect\tblForFigNine
\label{@FIG.09}}
\end{figure}

The bang-bang extremal $r^-{\to}r''{\to}r^+$ (blue curve) in Fig.~\ref{@FIG.02}a provides another example of the loop trap that is also the topological trap relative to type II optimal trajectory $r^-{\to}r^+$ (the specific parameters used in this example are: $\umax{=}\frac12$, $r^{-}{=}r_0{\propto}\{0,1,{-}\frac12\}$, $r^{+}{=}o{\propto}\{0,1,1\}$). In general, once the endpoints $r^-$ and $r^+$ satisfy the conditions of Proposition~\ref{*proposition:equatorial_always_type_II}, the time-optimal solution remains the same type II trajectory even in the limit $\umax{\to}0$, where the most  time optimal trajectories are of type I (see Proposition~\ref{*proposition-psi-criteria_of_type_I_II_solutions} and Fig.~\ref{@FIG.08}). Moreover the traps of the shown form will exist for any value of $\umax{<}{\sqrt{4{-}{(r^{-}_z{+}r^{+}_z)^2}}}/{|r^-_z-r^+_z|}$.

Another generic example of the traps of all three types can be straightforwardly constructed in the case $\umax{\gg}1$ (see Fig.~\ref{@FIG.10}) by selecting $o{\propto}\{1,0,\umax\}$ and choosing the initial state in vicinity of $z{=}1$: $r_0{\propto}\{c_1,c_2,\umax\}$, where $0{<}c_1{<}1$ and $c_2$ is any sufficiently small number. Although the vast majority of time-optimal solutions are of type II in the limit $\umax{\to}\infty$ (see Proposition~\ref{*proposition-psi-criteria_of_type_I_II_solutions}), for this special choice the optimal solution is of type I for any finite value of $\umax$ whereas the complementary type II extremal represents the topological trap. In the case $c_2{<}0$, there also exist a deadlock trap structurally similar to the ones shown in Fig.~\ref{@FIG.09}.

\newcommand{\tblForFigTen}{
\setlength{\tabcolsep}{0pt}
\begin{minipage}{\columnwidth}
\begin{tabular}
{|c|c|c|c|c|c|}
\hline \rowcolor{gray!30}
~extremal ~ & type &~$\ncp$ ~&~$\tilde{\Delta\tau}_1$ ~&~$\tilde{\Delta\tau}_2$ ~&~$\tilde{\Delta\tau}_3$ ~\\
\hline deadlock trap    & I   &    0 & 0.020  &   -   &  -    \\
\hline optimal solution & I   &    2 & 0.0327 & 0.262 & 0.017 \\
\hline topological trap & II  &    2 & 0.075  & 0.031 & 0.324 \\
\hline
\end{tabular}
\end{minipage}
}

\begin{figure}[t!]
   \includegraphics[width=\columnwidth]{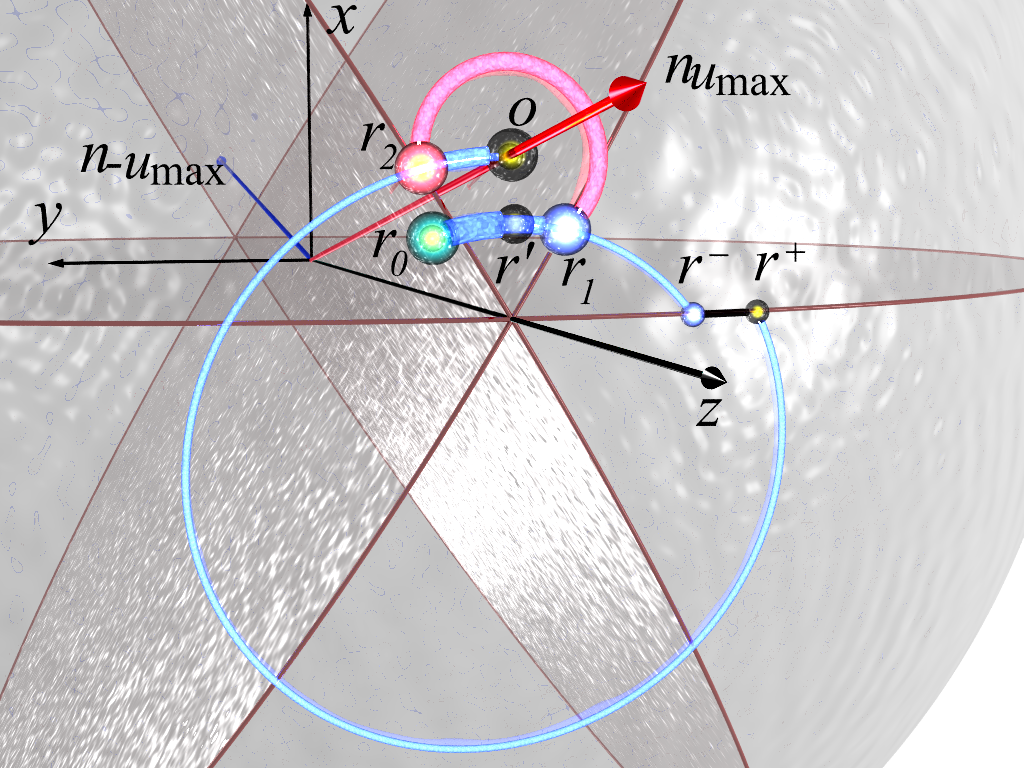}
\caption{The optimal solution (medium-thick trajectory $r_0{\to}r_1{\to}r_2{\to}o$), topological trap (thin trajectory $r_0{\to}r^-{\to}r^+{\to}o$) and deadlock trap (thick trajectory $r_0{\to}r'$) for the time-optimal control problem \eqref{01.-performance_index},\eqref{01.-finite_domain},\eqref{01.-finite_time} with $\umax{=}8$, $r(0){\propto}\{\frac12,\frac12,\umax\}$, $o{\propto}\{1,0,\umax\}$. The segments colored blue$\backslash$black$\backslash$red correspond to $u(\tau){=}{-}\umax\backslash0\backslash{+}\umax$ and are associated with rotations about the axes $\vec n_{-\umax}\backslash\vec\epsilon_x\backslash\vec n_{-\umax}$. The durations $\Delta\tau_i$ of the consequent bang arcs are summarized in the table:\\
\protect\tblForFigTen
\label{@FIG.10}}
\end{figure}

These observations lead to the following key proposition:
\begin{proposition}\label{*proposition-Traps_Never_Die(time-optimal)}
For any value of $\umax$ there exist initial states $\rho_0$, observables $\hat O$ and locally time-optimal control policy $\tilde u(\tau)$ whose constitutes the non-simple traps of time optimal control problem \eqref{01.-performance_index},\eqref{01.-finite_domain},\eqref{01.-finite_time}.
\end{proposition}

\section{Traps in the fixed-time optimal control\label{@SEC:traps(T=const)}}

Consider the problem \eqref{01.-performance_index}, \eqref{01.-finite_domain} where the control time $T$ is fixed. Specifically, we will be interested in the case
\begin{gather}
T{=}\mbox{const}{\gg}\frac{\pi^2}{\alpha}
\end{gather} when the kinematically optimal solutions exist for any given $\rho_0$ and $\hat O$. We again will exclude the class of perfect loop traps from the analysis for the same reasons as in the previous section. Intuitively one can expect that the probability of trapping in the local extrema (other than perfect loops) should be small at large $T$. However, it is not clear if there exists such value of $T$ that the functional \eqref{01.-performance_index} will become completely free of such traps.

To answer this question, note that in line with the analysis given in Sec.~\ref{@SEC:preliminaries} any trap should be represented by either type I or type II extremal. However, the maximal number of switchings is no longer limited by inequalities similar to Proposition~\ref{*proposition:n<=pi/alpha}. At the same time, Proposition~\ref{*proposition:corner_points_locations} remains applicable (see \emph{Remark 1} in Appendix~\ref{@APP:4}). Recall that its proof is based on introduction of the ``sliding'' variations $\delta\gamma_i$ which shift the angular positions of the ``images'' of corner points on the diagram of Fig.~\ref{@FIG.05} (see Appendix~\ref{@APP:4}). The explicit expression for the ``sliding'' variation around the $i$-th corner point up to the third order in the associated control time change $\delta\tau_i$ is given by eq.~\eqref{app4.-sliding_variation-second-order-Tailor_expansion}. By definition, if the trajectory $\tilde u(\tau)$ is type I trap, then no admissible control variation $\delta u$ can improve the performance index \eqref{01.-performance_index}. Consider the subset $\Omega$ of such variations composed of infinitesimal sliding variations $\delta\gamma_i$ that preserve the total control time $T$. Then, the necessary condition of trap $\tilde u(\tau)$ is absence of the non-uniform sliding variation $\delta u(\tau){\in}\Omega$ that leaves the trajectory endpoint $r_{\ncp{+}1}$ intact. Indeed, the trajectory associated with varied control $\tilde u{+}\delta u$ would deliver the same value of the performance index but at the same time is not the locally optimal solution (since it is no longer the type I extremal) which implies that $\tilde u$ is not locally optimal.

Using \eqref{app4.-sliding_variation-second-order-Tailor_expansion} the stated necessary condition can be rewritten as the requirement of definite signature of the quadratic form \eqref{app4.-Q-form}, where the parameters $q_i$ were introduced in Proposition \eqref{*proposition:second-order_Q-criterion}. The necessary condition of the sign definiteness is that all (probably except one) parameters $q_i$ are either non-positive or non-negative. Using Fig.~\ref{@FIG.05} one can see that in the case of long $T$ only the second option can be realized with $\eta{\simeq}0$, $\eta{\simeq}{-}\frac{\pi}2$ and $\eta{\simeq}{-}\frac{\pi}{3}$ (the case $\eta{\simeq}{{-}\pi}$ must be eliminated  because it implies $\umax{=}0$). One can show that the last two variants lead to saddle points rather that to the local extrema. The remaining case $\eta{\simeq}0$ leaves the two options $\theta{\simeq}0$ and $\theta{\simeq}2\pi$. The last option corresponds to positive constant $c_{i,1}$ in \eqref{02.-[rho,O]}, which indicates the possibility of increasing  $J$ via monotonic ``stretching'' the time: $T{\to}T{+}\delta T(T)$,  $u(\tau){\to}u(\tau{-}\delta T(\tau))$, where $\delta T(\tau)$ is an infinitesimal positive monotonically increasing function. At the same time, the associated parameters $q_i$ are all negative, so there exists the combination of  variations $\delta\tau$ of arcs durations $\Delta\tau$ which will result in achieving the same value of the performance index at shorter time. Thus, we can conclude that it is also possible to increase $J$ at fixed time T via proper combination of these two variations, so the variant $\theta{\simeq}2\pi$ should be dismissed as a saddle point. Only the remaining choice $\theta{\simeq}0$ is consistent with an arbitrary number of $q_i$ of the same sign. However, in this case the length of each bang arc also reduces to zero. As result, the maximal duration of such optimal trajectories is limited by the inequality $T{\lesssim}\pi$.

This analysis leads us to remarkable conclusion:
\begin{proposition}\label{*proposition-complex_traps_die(time-fixed)}
The fixed-time optimal control problem \eqref{01.-performance_index} is free of non-simple traps for sufficiently long control times $T$.
\end{proposition}

The spirit of this conclusion is in line with the results of numerical simulations performed in \cite{2012-Pechen}. With this, it is worth recalling that the general time-fixed problem may have a variety of perfect loop traps for any value of $\umax$ and, thus, is not trap-free in the strict sense. These traps were missed in the simulations in \cite{2012-Pechen} due to the specifics of numerical optimization procedure.

\section{Summary and conclusion\label{@SEC:conclusion}}

All stationary points of the time optimal control problem and all saddles and local extrema of the fixed-time optimal control problem are represented by the piecewise-constant controls of types I and II sketched in Fig.~\ref{@FIG.01} (the associated characteristic trajectories $\rho(\tau)$ on the Bloch sphere are shown in Figs.~\ref{@FIG.03} and \ref{@FIG.11}, correspondingly). 
We systematically explored the anatomy of stationary points of each type. Specifically, we identified
the locations and relative arrangements of corner points on the Bloch sphere (propositions~
\ref{*proposition:singular_sections_on_one_xz_side},
\ref{*proposition:s*r_x*r_y<0},
\ref{*proposition:corner_points_locations},
\ref{*proposition:second-order_Q-criterion},
\ref{*proposition-n_max_for_time-optimal_case}, 
\ref{*proposition-r_i-global_optimality_bounds})
and estimated their total number (propositions~
\ref{*proposition:n<=pi/alpha},
\ref{*proposition:n<=2_for_u>sqrt(1+sqrt(2))},
\ref{*proposition-n_max_for_time-optimal_case},
\ref{*proposition-n_max_for_time-optimal_case_refined},
\ref{*proposition-n_min-bounds}). These characteristics together with propositions \ref{*proposition:equatorial_always_type_II},
\ref{*proposition:type_II_no_interior_bangs},
\ref{*proposition-0^II_trajectories_are_locally_time-optimal} and
\ref{*proposition-psi-criteria_of_type_I_II_solutions} allow to determine whether the given extremal is a saddle point or locally optimal solution, and also to predict the shape of globally optimal solution. The presented results (except Proposition \ref{*proposition:n<=pi/alpha}) substantially generalize and refine the estimates obtained in previous studies \cite{2006-Boscain,2013-Hegerfeldt}.
Moreover, this study, to our knowledge, is the first example of a systematic analytic exploration of the overall topology of the quantum landscape $J[u]$ in the presence of constraints on the control $u$ and for the arbitrary initial quantum state $\rho_0$ and observable $\hat O$.
In particular, we distinguished 4 categories of traps tentatively called deadlock, topological, loop and perfect loop traps. The landscape can contain an infinite number of perfect loops whereas the number of traps of other types is always finite. Among them the number of deadlock traps and loops decreases with increasing value of the constraint $\umax$ in eq.~\eqref{01.-finite_time}. Nevertheless, we have shown by an explicit example that the traps of all categories can simultaneously complicate the landscape $J[u]$ of the time-optimal control problem regardless of the value of $\umax$. So, this is the case where the intuitive attempt to ``extrapolate'' the conclusions based on analysis of the case of unconstrained controls totally fails.

The fixed-time control problem is more intriguing. On one hand we formally showed that it is impossible to completely ``flatten'' all the traps in this case by increasing the value of $\umax$. This result is in line with generic experience concerning the optimal control in technical applications. However, if the control time is long enough (specifically, if $T{\gg}\pi^2/\arctan{\umax}$) the only traps which can survive are perfect loops. These traps can be easily escaped via simple modification of any gradient search algorithm at virtually no computational cost. Thus, the quantum landscape appears as trap-free from practical perspective, which supports the common viewpoint in quantum optimal control community. Since the controlled two-level system is the benchmark for quantum information processing this finding is relevant for efficient optimal control synthesis in a variety of experiments on cold atoms, Bose-Einstein condensates, superconducting qubits etc.

The key methodological feature of the presented derivations is introduction of the sliding variations which makes it possible to extensively rely on highly visual and intuitive geometrical arguments. For this reason, we believe that the mathematical aspect of the paper constitutes instructive introduction into high-order analysis of optimal processes.


\appendix

\section{Proof of Proposition~\texorpdfstring{\ref{*proposition:equatorial_always_type_II}}{Lg}\label{@APP:1}}
Here we consider the case $r^-_y,r^+_y{>}0$. The case $r^-_y,r^+_y{<}0$ can be treated similarly. Simple geometrical analysis leads to the following expression for the travel time difference $\delta T$ between bang-bang (orange) and ``equatorial'' (black) trajectories shown in Fig.~\ref{@FIG.02}a:
\begin{gather}
\delta T_{\idx{a}}{=}\cos (\alpha ) \left(\arcsin\left(\frac{\frac{\delta_z}{2} \sec (\alpha )-\cos (\alpha ) {r^+_z}}{\sqrt{1{-}\sin ^2(\alpha ) {r^+_z}^2}}\right)\right.+\notag\\
\arcsin(\frac{\cos (\alpha ) {r^+_z}}{\sqrt{1{-}\sin ^2(\alpha ) {r^+_z}^2}})
{-}\arcsin(\frac{\cos (\alpha ) {r^-_z}}{\sqrt{1{-}\sin ^2(\alpha ) {r^-_z}^2}}){+}\notag\\
\left.\arcsin(\frac{\frac{\delta_z}{2} {} \sec (\alpha ){+}\cos (\alpha ) {r^-_z}}{\sqrt{1-\sin ^2(\alpha ) {r^-_z}^2}})\right){-}
\arcsin({r^+_z}){+}\arcsin({r^-_z}),
\end{gather}
where $\delta_z{=}r^+_z-r^-_z$. Let us fix one of the endpoints $r^{\pm}$ and vary the position of another one. Note that $\delta T_{\idx{a}}|_{\delta_z{=}0}{=}0$ for any admissible value of $r^{\pm}_z$. Furthermore,
\begin{gather}\label{App.1.-equatorial_variation}
{\pm}\der{\delta T_{\idx{a}}}{r^{\pm}_z}{=}\frac{(1{-}{r^{\pm}_z}^2) \left(\sqrt{1{-}\frac{{r^{\pm}_z} \delta _z}{1{-}{r^{\pm}_z}^2}}{-}\sqrt{1{-}\frac{{r^{\pm}_z} \delta _z{+}\frac{\delta _z^2}{4} \sec ^2\alpha }{1{-}{r^{\pm}_z}^2}}\right)}
{\left(\csc ^2\alpha{-}{r^{\pm}_z}^2\right) \sqrt{1{-}{r^{\pm}_z} \delta _z{-}{r^{\pm}_z}^2{-}\frac{\delta _z^2}{4} \sec ^2\alpha }}{>}0.
\end{gather}
This allows to conclude that $\delta TT_{\idx{a}}{>}0$ for any $\delta_z{>}0$ which finishes the proof of Proposition for the case $r^-_yr^+_y{>}0$.

Consider now the case $r^-_yr^+_y{<}0$. For clarity, we will assume that $r^-_y{>}0$ $r^-_z{<}r^+_z$ (see Fig.~\ref{@FIG.02}b). The remaining cases can be analyzed similarly. The time difference $\delta T_{\idx{b}}$ between ``equatorial'' (black) and the green trajectories and its derivative with respect to the position of the endpoint $r^+_z$ read:
\begin{gather}
\delta T_{\idx{b}}{=}\arccos(r^+_z){-}\cos (\alpha )\arccos\left(\frac{r^+_z \cos (\alpha )}{\sqrt{1-{r^+_z}^2 \sin ^2(\alpha )}}\right);
\end{gather}
\begin{gather}
\pder{}{r^+_z}\delta T_{\idx{b}}{=}{-}\frac{2 \sqrt{1{-}{r^+_z}^2} \sin ^2(\alpha )}{{r^+_z}^2 \cos (2 \alpha ){-}{r^+_z}^2{+}2}.
\end{gather}
These expressions show that $\delta T_{\idx{b}}|_{r^+_z{=}1}=0$ and that $\pder{}{r^+_z}\delta T_{\idx{b}}{>}0$ for any admissible value of $r^+_z$. Thus, $\delta T_{\idx{b}}{>}0$ which proofs Proposition for the case $r^-_yr^+_y{<}0$.

\section{Proof of Proposition~\texorpdfstring{\ref{*proposition:s*r_x*r_y<0}}{Lg}\label{@APP:3}}
The proof is based on explicit construction of the second-order  McShane's (needle) variation of the control $\tilde u(\tau)$ which decreases $\tilde T$ if the inequality \eqref{04.-proposition:s*r_x*r_y<0} is violated. Choose arbitrary infinitesimal parameter
$\delta\tau^-{\to}0$ and denote $r_i^-{=}\tilde r(\tilde\tau_i{-}\delta\tau^-)$. Under assumptions of Proposition it is always possible (except for the trivial case $\tilde r_{i,y}{=}0$) to choose another small parameter $\delta\tau^{+}$ such that the state vector $r_i^+{=}\tilde r(\tilde\tau_i{+}\delta\tau^+)$ obeys the equality: $r_{i,x}^-{=}r_{i,x}^+$. It is evident that the Bloch vector $r_{i,x}^-$ can also reach $r_{i,x}^+$ in the course of free evolution with $u{=}0$ after certain time $\delta\tau^0$. If we require that $\delta \tau_i^+,\delta\tau_i^0|_{\delta\tau_i^-{\to}0}{=}0$ then both $\tau_i^+$ and $\tau_i^0$ are uniquely defined by $\delta\tau_i^-$:
\begin{gather}\label{app.3-variation_timings}
\delta\tau_i^+{=}\frac{\delta\tau_i^-(\tilde r_{i,y}{+}2 \delta\tau_i^{-}\tilde r_{i,z})}{\tilde r_{i,y}}+o({\delta\tau_i^-}^2);\notag\\
\delta\tau_i^0{=}\frac{2\delta\tau_i^-(\delta\tau_i^- (\tilde u_i^-  \tilde r_{i,x}{+}\tilde r_{i,z}){+}\tilde r_{i,y})}{\tilde r_{i,y}},
\end{gather}
and thus, $\delta\tau_i^0{-}\delta\tau^+{-}\delta\tau^-{=}{2 \tilde u_i^-{(\delta\tau_i^-)}^2 \tilde r_{i,x}}/{\tilde r_{i,y}}$. The latter quantity should be nonnegative for the locally time-optimal solution which leads to eq.~\eqref{04.-proposition:s*r_x*r_y<0}.

\section{Proof of Proposition~\texorpdfstring{\ref{*proposition:type_II_no_interior_bangs}}{Lg}\label{@APP:2}}
Consider any type $^s$II extremal with $s{>}0$. By definition, such extremals must contain at least one interior bang segment $\tau{\in}[\tilde\tau_i,\tilde\tau_{i{+}1}]$ of length $\tilde{\Delta\tau}_i{=}m{\pi}\cos(\alpha)$ ($m\in{\mathbb{N}},0{<}\tilde\tau_i,\tilde\tau_{i{+}1}{<}T$). Since $\tilde r_i{=}\tilde r_{i{+}1}$ both the value of the performance index $J$ and duration $T$ will not change if this segment will be ``translated'' in arbitrary new point $\tilde r(\tau_i'(\kappa))$ of extremal via the following continuous variation $\tilde u(\tau){\to}u(\kappa,\tau)$ (${-}{\tau_i}{<}\kappa{<}T{-}m\pi\cos\alpha$):
\begin{gather}
u(\kappa,\tau){=}
\begin{cases}
\tilde u(\tau), & \tau{<}\tilde \tau_i{+}\frac{\kappa{-}|\kappa|}{2}\lor \tau{>}\tilde\tau_{i{+}1}{+}\frac{\kappa{+}|\kappa|}{2};\\
\tilde u_i^+, & \tilde\tau_i{+}\kappa{<}\tau{<}\tilde\tau_{i{+}1}{+}\kappa;\\
u(\tau{-}\tilde{\Delta\tau}_i) & \mbox{otherwise,}
\end{cases}
\end{gather}
where $\tau'(\kappa){=}\tilde\tau_i{+}{\kappa}{+}\frac12(1{+}\frac{\kappa}{|\kappa|})\tilde{\Delta\tau}_i$. 

Suppose that $\tilde u(\tau)$ is locally time-optimal solution. Then all the family of control policies $\{u(\kappa, \tau),r(\kappa,\tau)\}$ should be locally time-optimal too. Since $s{>}0$ it is always possible to select the value $\kappa{=}\kappa_0$ such that $\tilde r(\tau'(\kappa_0))$ is interior point of the bang arc with $\tilde u(\tau'(\kappa_0)){=}{-}\tilde u_i^+$ and $\tilde r_x(\tau'(\kappa_0)){\ne}0$. However, the resulting trajectory $r(\kappa_0,\tau)$ is both $\Lambda$- and $V$-shaped in the neighborhood of point $r(\tilde \tau_i{+}\kappa_0){=}r(\tilde\tau_{i{+}1}{+}\kappa_0)$. According to Proposition \eqref{*proposition:s*r_x*r_y<0} such trajectory can not be time-optimal. The obtained contradiction finishes the proof.

\section{Proof of Proposition~\texorpdfstring{\ref{*proposition-0^II_trajectories_are_locally_time-optimal}}{Lg}\label{@APP:9}}

Let $u'(\tau)$ be the control strategy obtained via arbitrary McShane variation $\delta u(\tau)$ of the control $\tilde u(\tau)$. Let us show that $u'(\tau)$ is less time efficient than some member $u''(\tau)$ of the control family ${\cal F}^{[\mathrm k]}({u''}^{\idx{anz}}(\tau))$ with the same $\mathrm k$ but perhaps the different anzatz $\tilde {u''}^{\idx{anz}}$. For this we will need the following lemma which is complementary to Propositions~\ref{*proposition:equatorial_always_type_II} and \ref{*proposition:s*r_x*r_y<0}:
\begin{figure}[t!]
   \includegraphics[width=0.35\textwidth]{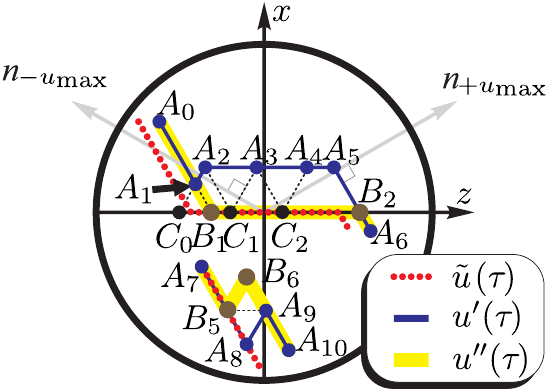}
\caption{Projections of the characteristic pieces of the original, varied and reduced trajectories $\tilde r(\tau)$, $r'(\tau)$ and $r''(\tau)$ on the $xz$ plane (it is assumed that $y$-components of all shown parts of trajectories are greater than zero). The color associations are indicated in the inset.
\label{@FIG.12}}
\end{figure}

\begin{llemma}\label{*lemma-antiequatorial_optimality}
Suppose that $r'(\tau')$ is junction point of two bang arcs of the trajectory $u(\tau)$ such that $r'_x{=}0$. Consider any two points $r^-(\tau^-)$ and $r^+(\tau^+)$ ($\tau^-{<}\tau'{<}\tau^+$) on adjacent arcs such that $r^-_x{=}r^+_x$ and the complete segment $r'_yr_y(\tau){>}0$ for any $\tau{\in}(\tau^-,\tau^+)$. Denote as $\widehat{\Delta\tau}$ the minimal duration of free evolution ($u{=}0$) required to reach $r^+$ starting from $r^-$. Then, $\widehat{\Delta\tau}{>}\tau^+{-}\tau^-$.
\end{llemma}

Since $\tilde u_{\idx{anz}}(\tau)$ is locally optimal by assumption it is sufficient to consider the variations of the $\delta u(\tau)$ which do not involve the vicinities of the trajectory endpoints. Moreover, it is sufficient to analyze the variations $\delta u(\tau)$ which are nonzero only in vicinities points where $\tilde r(\tau){=}0$. To show this consider the McShane variation in the arbitrary interior point $A_8$ of the bang arc (see Fig.~\ref{@FIG.12}). Consider the piece $A_7A_8A_9A_10$ of the varied trajectory $r'(\tau)$. According to Proposition~\ref{*proposition:s*r_x*r_y<0} (see eq.~\eqref{app.3-variation_timings}) the path $B_5B_6A_9$ is more time-efficient than $B_5A_8A_9$ if the varied segment $A_8A_9$ is sufficiently small. Thus, the trajectory $A_7B_5B_6A_10$ is more time-efficient than original segment $A_7A_8A_9A_10$. By repeated application of the same reasoning to the modified pieces of trajectory one can replace the control $u'(\tau)$ with the more effective strategy which differs from $\tilde u(\tau)$ only in vicinities of the points $r'$ with $r'_x{\to}0$. Since it is sufficient to consider only this modified control policy we will rename it as $u'(\tau)$ and will refer as the initial variation in the subsequent analysis.

The characteristic piece $A_0A_2A_5A_6$ of the resultant trajectory is shown in Fig.~\ref{@FIG.12}. Following the proof of Proposition~\ref{*proposition:equatorial_always_type_II} (see eq.~\eqref{App.1.-equatorial_variation}) the path $C_0A_2C_1$ is less time-efficient than $C_0A_1B_1C_1$. This implies that the path $A_1B_1C_1$ is more time-efficient than $A_1A_2C_1$. According to Lemma~\ref{*lemma-antiequatorial_optimality}, the path $A_2C_1A_3$ is more time-efficient than the path $A_2A_3$ associated with the free evolution. As a result, the trajectory segment $A_1A_2A_3$ of the $r'(\tau)$ is less time efficient than the combination of the segment $A_1B_1C_1$ of the trajectory $r''(\tau)$ with the segment $C_1A_3$. By continuing the similar analysis one finally comes to conclusion that the part of trajectory $r'(\tau)$ between the points $A_0$ and $A_5$ is less time efficient than the corresponding segment of $u''(\tau)$.
Applying the same reasoning to the entire trajectory $r'(\tau)$ we will reduce the original variation to the $^0II$ type control $u''(\tau)$ and trajectory $r''(\tau)$. Note that we must assume that all the singular segments where $u''(\tau){=}0$ are located on the same side with respect to $xz$ plane (otherwise the control time can be further reduced by eliminating some singular segments following the proof of Proposition~\ref{*proposition:singular_sections_on_one_xz_side}, see eq.~\eqref{app.3-variation_timings}). This mean, that all the interior bang sections of the control $u''(\tau)$ are of length $m\pi/cos(\alpha)$ $(m\in\mathbb{N})$. Thus, the trajectory $u''(\tau){=}0$ must belong to the family ${\cal F}^[\mathrm k]({u''}^{\idx{anz}}(\tau))$ with  the same index $\mathrm k$ as ${\cal F}^[\mathrm k]({u''}^{\idx{anz}}(\tau))$ and the anzatz ${u''}^{\idx{anz}}(\tau)$ related to $\tilde{u}^{\idx{anz}}(\tau)$ via infinitesimal variation. Since $\tilde{u}^{\idx{anz}}(\tau)$ is time-optimal the performances and control times associated with policies ${u''}^{\idx{anz}}$ and $\tilde u^{\idx{anz}}$ are related as $\tilde J^{\idx{anz}}{\geq}{J''}^{\idx{anz}}$ and $\tilde T^{\idx{anz}}{\leq}{T''}^{\idx{anz}}$. Consequently,  $\tilde J{\geq}J''$ and $\tilde T{\leq}T''$, so that the control policies $u''(\tau)$ and $u'(\tau)$ can not be more effective than $\tilde u(\tau)$. The latter conclusion completes the proof of Proposition~\ref{*proposition-0^II_trajectories_are_locally_time-optimal}.

\begin{proof}[Proof of the Lemma~\ref{*lemma-antiequatorial_optimality}]
For concreteness, consider the case $r'_y{>}0$, $r^-_x{>}0$. Denote $\widehat{\delta\tau}{=}r^+(\tau^+){-}r^+(\tau^+){-}\widehat{\Delta\tau}$. Using simple geometrical considerations one can find that
\begin{gather}
\widehat{\delta\tau}(r^-_x,r'_z){=}\frac{1}{2} \sum_{s{=}\pm1}\left(\arcsin\left(\frac{s r'_z{-}r^-_x \cot (\alpha )}{\sqrt{1{-}{r^-_x}^2}}\right){+}\right.\notag\\
\left.\frac{\arcsin\left(\frac{r^-_x \csc (\alpha ){-}sr'_z \cos (\alpha )}{\sqrt{1{-}{r'_z}^2 \sin^2(\alpha )}}\right)}{\sqrt{\tan ^2(\alpha ){+}1}}\right).\label{app.9.-delta_tau}
\end{gather}
By differentiating \eqref{app.9.-delta_tau} we find that $\pder{}{r^-_x}\widehat{\delta\tau}(r^-_x,r'_z{=}0){=}{-}\frac{x^2 \sin (2 \alpha ) \sqrt{1{-}x^2 \csc ^2(\alpha )}}{\left(x^2{-}1\right) \left(\cos (2 \alpha ){+}2 x^2{-}1\right)}{<}0$ for any admissible $r^-_x{>}0$. Similarly, one can show that $\widehat{\delta\tau}(r^-_x{=}0,r'_z){=}0$ and $r'_z\pder{}{r'_z}\widehat{\delta\tau}(r^-_x,r'_z){<}0$ for any admissible $r'_z{\ne}0$. Taken together, these relations lead to conclusion that $\widehat{\delta\tau}(r^-_x,r'_z){<}0$ for any admissible $r^-_x{>}0$ which completes the proof for the case $r'_y{>}0$, $r^-_x{>}0$. Other cases can be analyzed in the same way.
\end{proof}

\section{Proof of Proposition~\texorpdfstring{\ref{*proposition:corner_points_locations} and \ref{*proposition:second-order_Q-criterion}}{Lg}\label{@APP:4}}
One can directly check that the transformation ${\cal S_{\pm}}{=}\exp(\tilde{\Delta\tau}{\cal L}(\pm\umax))$ is equivalent to the composition of rotation ${\cal S}_{\vec \epsilon_z}(\mp\xi)$ around axis $\vec \epsilon_z$ by angle $\mp\xi$ with rotation ${\cal S}_{\vec n_{\pm\umax}}(\eta)$ around the normal vector $\vec n_{\pm\umax}$ to the plane $\lambda_{\pm\umax}$ by $\eta$:
\begin{gather}\label{04.-S+-decomposition}
{\cal S_{\pm}}{=}{\cal S}_{\vec n_{\pm1}}(\eta){\cal S}_{\vec \epsilon_z}(\mp\xi)~~~({-}\pi{<}\eta{<}0;~~0{<}\xi{<}\pi),
\end{gather}
where the domain restrictions on the values of $\eta$ and $\xi$ result from \eqref{03.-regular_segment_length}. Thus, the state transformation induced by any two subsequent bang arcs is equivalent to rotation around $\vec n_{\pm\umax}$ by angle $2\eta$. This proofs that the all odd (even) corner points are situated in the same plane orthogonal to $\vec n_{u_{1}^-}$ ($\vec n_{u_{1}^+}$) and parallel to $\vec\epsilon_z$. More specifically, they are located on the circles $\vec{r}\vec n_{\pm\umax}{=}c_0$ which are mirror images of each other in $xz$ plane.

In order to complete proof of Proposition \ref{*proposition:corner_points_locations} it remains to show that $\vec\epsilon_z{\in}\lambda_{\pm\umax}$ (i.e. that $c_0{=}0$). Since it is already shown that $\vec\epsilon_z{\parallel}\lambda_{\pm\umax}$ it is enough to prove that there exist an least one common point with axis $\vec\epsilon_z$. Consider the infinitesimal variations $\delta \tau_i^-$ and $\delta \tau_i^+$ of the durations $\tilde{\Delta\tau}_i$ and $\tilde{\Delta \tau}_{i{+}1}$ of the bang arcs adjacent to arbitrary corner point $\tilde r_i{=}\tilde r(\tilde\tau_i)$, such that the transformation ${\cal S}{=}\exp(\delta\tau_i^-{\cal L}(\tilde u_i^-))\exp(\delta\tau_i^+{\cal L}(\tilde u_i^+))$ moves the point $\tilde r_i$ into $r'_i{\in}\lambda_{\tilde u_{i}^-}$. In other words, we require that $\tilde r_i$ and $r'_i$ should relate by infinitesimal rotation ${\cal S}_{\vec n_{\tilde u_{i}^-}}(\delta\gamma_i)$. For convenience, we will call such variations as ''sliding'' ones. The form of decomposition \eqref{04.-S+-decomposition} indicates that the sliding variation at $r_i$ shifts the locations of all subsequent corner points $\tilde r_{j>i}{\to}r_j'$ by similar rotations
${\cal S}_{\vec n_{u_{j}^-}}(\delta\gamma_i)$ around the associated axes $\vec n_{u_{j}^-}$. Consider the arbitrary composition of the sliding variations, such that the trajectory start and end points remain fixed, i.e. $\sum_{i}{\delta\gamma_i}{=}0$. If the extremal $\tilde u$ is locally optimal then such variations should not allow the reduction of the control time $T$: $\sum_{i}\delta\tau_i{\leq}0$, where $\delta\tau_i{=}\delta \tau_i^-{+}\delta \tau_i^+$. This requirement leads to the following first-order (in $\delta\tau_i$) necessary optimality condition:
\begin{gather}\label{04.-corner_point_first_order_optimality}
\forall i,j: \der{\delta\gamma_i}{\delta \tau_i}{=}\der{\delta\gamma_j}{\delta \tau_j}
\end{gather}
Using simple geometrical analysis it is possible to explicitly calculate the derivatives in \eqref{04.-corner_point_first_order_optimality}:
\begin{gather}
\der{\delta\gamma_i}{\delta \tau_i}{=}\frac{2\sqrt{\cos ^2\left(\frac{\theta }{2}\right){+}\umax^2}}{\frac{\tilde r_{i,x}}{\tilde r_{i,y}}\tilde u_i^{-}\sin \left(\frac{\theta }{2}\right){-}\sqrt{1{+}\umax^2}\cos \left(\frac{\theta }{2}\right)}.
\end{gather}
We can conclude that equalities \eqref{04.-corner_point_first_order_optimality} are equivalent to condition: $\frac{\tilde r_{i,x}}{\tilde r_{i,y}}\tilde u_i^-{=}\mathrm{const}$ which directly leads to conclusion that $\vec\epsilon_z{\in}\lambda_{\pm1}$ and completes the proof of Proposition~\ref{*proposition:corner_points_locations}.

\emph{Remark 1.} It is worth stressing that the above proof of Proposition~\ref{*proposition:corner_points_locations} does not explicitly depend on the time optimality of the trajectory $\tilde u(\tau)$. Thus, its statement is generally valid for any type I extremal locally optimal with respect to small variations of control $\tilde u(\tau)$, including the case of fixed control time $T$.

The proof of Proposition~\ref{*proposition:second-order_Q-criterion} follows from the analysis of the higher-order terms in sliding variation along the extremal trajectory. Calculations result in the following expression:
\begin{gather}\label{app4.-sliding_variation-second-order-Tailor_expansion}
\delta\gamma_i{=}2\cos(\frac{\xi }{2})\delta\tau_i{-}\left|\frac{\sin^3(\frac{\xi }{2})}{\umax }\right| q_i\delta\tau_i^2{+}q_i^{(3)}\delta\tau_i^3{+}o(\delta\tau_i^3),
\end{gather}
where
\begin{align}
q_i^{(3)}{=}&\frac{1}{3}\umax^2\cos\left(\frac{\xi }{2}\right)\left[2\sec^2\left(\frac{\eta }{2}\right){-}3 q_i^2 \tan ^4\left(\frac{\eta }{2}\right){-}\right.\notag\\
&\left.6 \cot (\gamma_i ) \left(\tan \left(\frac{\eta }{2}\right){+}(q_i{+}1) \tan ^3\left(\frac{\eta }{2}\right)\right)\right]\label{app4.-sliding_variation-third-order_term}
\end{align}
The necessary condition of the local optimality is thus the inequality $\sum_{i{=}1}^nq_i\delta\tau_i^2{\geq}0$ in which the variations $\delta\tau_i$ are subject to constraint $\sum_{i{=}1}^n \delta\tau_i{=}0$. The power of sliding variation is in the fact that the quadratic form in the left-hand side of this inequality is diagonal (i.e{.} the contributions of the sliding variations $\delta\gamma_i$ are independent up to the second order in $\delta\tau_i$). Thus, optimality implies non-negativity of the following simple quadratic form:
\begin{gather}\label{app4.-Q-form}
Q_{kj}{=}\delta_{kj}q_k{+}q_n~~~(k,j{=}1,...,n{-}1),
\end{gather}
which can be easily rewritten in the form of statement of Proposition~\ref{*proposition:second-order_Q-criterion}.

\section{Proof of Proposition~\texorpdfstring{\ref{*proposition:n<=2_for_u>sqrt(1+sqrt(2))}}{Lg}\label{@APP:5}}
Let $q'{=}q_{i'}{<}0$ be the smallest term in the set $\{q_i\}$. By applying Proposition~\ref{*proposition:second-order_Q-criterion} to the corner points adjacent to ${i'}$-th we have: $q_{{i'}{\pm}1}{+}q_{{i'}}{<}0$. These inequalities can be rewritten after some algebra as:
\begin{gather}
\delta\gamma_{{i'}}{>}{-}\frac{\eta }{2}{-}\arccos\left(\sqrt{\sin ^2\left(\frac{\eta }{2}\right)(\cos (\eta)+2)}\right);\notag\\
\delta\gamma_{{i'}}{<}\frac{\eta }{2}+\cos ^{-1}\left(\sqrt{\sin ^2\left(\frac{\eta }{2}\right) (\cos (\eta )+2)}\right)\label{04.-corner_point_second_order_optimality_refined}
\end{gather}
where $\delta\gamma_{i'}{=}(\gamma_{i'} \mod \pi){-}\frac{\pi}{2}$ ($|\delta\gamma_{{i'}}|{<}\frac{\pi{+}\eta}{2}$). One can show that at least one of the inequalities  \eqref{04.-corner_point_second_order_optimality_refined} holds if $|\eta|{<}\arccos(\sqrt{2}{-}1)$. From the definition of $\eta$ it follows that the latter inequality holds for any $\umax{>}\sqrt{1{+}\sqrt{2}}$. This means that for this range of controls the $i'$-th corner point can be only either the left-most or the right-most corner point of time-optimal extremal. Using Fig.~\ref{@FIG.05} one can accordingly improve the estimate for $n_{\mathrm{max}}$: $n_{\mathrm{max}}{\leq}\left[\frac{|\eta|}{\pi{-}|\eta|}{+}2\right]{\leq}2$ for $\umax{>}\sqrt{1{+}\sqrt{2}}$ Q.E.D.

\section{Proof of Proposition ~\texorpdfstring{\ref{*proposition-n_max_for_time-optimal_case}}{Lg}\label{@APP:6}}

Suppose that $\tilde r_{i'}$ is interior corner point of the globally time-optimal solution. From \eqref{04.-switching_points_r_i-explicit} it follows that $\tilde r_{i,x}{=}\frac{|\tilde u_{1}^+|}{\tilde u_{1}^+}\sin (\zeta_i)\sin(\frac{\xi }{2}){\propto}c \sin (\zeta_i)$. where $c$ is some real constant. Since $|\sin(\zeta_{i'})|{<}\sin(\frac{\pi{+}\eta}{2})$ and $|\zeta_{i'}{-}\zeta_{i'{\pm}1}|{=}\frac{\pi{+}\eta}{2}$ the following inequality holds
\begin{gather}\label{04.-time-optimal_corner_point_condition}
\frac{\tilde r_{i',x}{-}\tilde r_{i'{\pm}1,x}}{\tilde r_{i',x}}{>}0.
\end{gather}
Proposition~\ref{*proposition:s*r_x*r_y<0} states that the trajectory curve in vicinity of $\tilde r_{i',x}$ should be $\lambda$-shaped ($V$-shaped) in the case of $\tilde r_{i',x}{<}0$ ($\tilde r_{i',x}{>}0$), as shown in Fig~\ref{@FIG.06}. Together with \eqref{04.-time-optimal_corner_point_condition} this means that both left and right adjacent arcs intersect the plane $x{=}\tilde r_{i',x}$ twice and have the second common point $\{\tilde r_{i',x},{-}\tilde r_{i',y},\tilde r_{i',z}\}$. However, the globally time optimal trajectories can not have intersections with themselves. This contradiction proves the statement of Proposition. The associated maximal number of switchings can be directly deduced using Fig.~\ref{@FIG.05}.

\begin{figure}[t]
\includegraphics[width=0.5\linewidth]{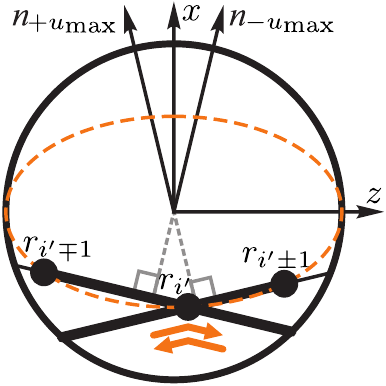}
\caption{Projection of the extremal in vicinity of the corner point $\tilde r_{i'}$ on $xz$-plane in the case $\tilde r_{i',x}{<}0$. Orange dashed ellipse is the projection of intersection of the Bloch sphere with the planes $\lambda_{\pm1}$. Arrows indicates the admissible routes of passing the point $\tilde r_{i'}$ according to Proposition~\ref{*proposition:s*r_x*r_y<0}.\label{@FIG.06}}
\end{figure}

\section{Proof of Proposition~\texorpdfstring{\ref{*proposition-r_i-global_optimality_bounds}}{Lg}\label{@APP:7}}
The statement of Proposition will be proven by contradiction. Suppose that the first of inequalities \eqref{*proposition-r_i-global_optimality_bounds_inequality} is violated (the case of violation of the second inequality can be treated similarly), i.e. $\exists i: ( \forall j :\tilde r_{i,x}{\leq}\tilde r_{j,x} \land \tilde r_{i,x}{<}0)$ . Using Proposition \ref{*proposition:s*r_x*r_y<0} we conclude that $\tilde r_{i,x}{\leq}\tilde r_{i{-}1,x},\tilde r_{i{+}1,x}$ and that the trajectory around $\tilde r_i$ is $\Lambda$-shaped: $\exists \epsilon,\forall \delta\tau{\in}({-}\epsilon,\epsilon):\tilde r_x(\tilde \tau_i{+}\delta\tau){<}\tilde r_x(\tau_i)$. Similarly to the proof of Proposition \ref{*proposition-n_max_for_time-optimal_case}, these observations mean that the both arcs $\tau{\in}(\tilde \tau_{i-1},\tilde \tau_i)$ and $\tau{\in}(\tilde \tau_i,\tilde \tau_{i+1})$ should cross the plane $x{=}\tilde r_{i,x}$ twice and thus have the common point $\{\tilde r_{i,x},{-}\tilde r_{i,y},\tilde r_{i,z}\}$. However, the latter contradicts with the assumed global time optimality of the trajectory $r(\tau)$.

\section{Proof of Proposition~\texorpdfstring{\ref{*proposition-n_max_for_time-optimal_case_refined}}{Lg}\label{@APP:8}}
Similarly to $\rP$ and $\rM$, let us introduce the new notations $r_{\pm}{=}\frac{\tilde r_1{+}\tilde r_n}2{\pm}\Sign(|\tilde r_{1,x}|{-}|\tilde r_{n,x}|)\frac{\tilde r_1{-}\tilde r_n}2$ for the
first and the last corner points $\tilde r_1$ and $\tilde r_{n}$ of trajectory $\tilde r(\tau)$, so that $|\tilde r_{+,x}|{\geq}|\tilde r_{-,x}|$. Using Fig.~\ref{@FIG.05} we find that:
\begin{gather}\label{app8.-ncp(phi)}
\ncp_{\mathrm{I}}{=}\left|\frac{\zeta_+-\zeta_-}{\pi+\eta}\right|{+}1{=}
|\frac{\arcsin(\tilde r_{+,x}\phi){-}\arcsin(\tilde r_{-,x}\phi)|}{2\arctan(\umax\phi)}|{+}1,
\end{gather}
where $\phi{=}\frac{1}{\sin(\frac{\xi}{2})}$. Eq.~\eqref{app8.-ncp(phi)}
can be rewritten as:
\begin{gather}\label{app8.-ncp(phi)-integral_form}
\ncp_{\mathrm{I}}{=}\frac{\int_0^{\phi}\left|\frac{\tilde r_{+,x}}{\sqrt{1{-}\phi^2\tilde r^2_{+,x}}}{-}\frac{\tilde r_{-,x}}{\sqrt{1{-}\phi^2\tilde r^2_{-,x}}}\right|d\phi}
{\int_0^{\phi}(\frac{\umax}{1{+}\umax^2\phi^2})d\phi}{+}1.
\end{gather}
The integrands in the numerator and denominator of \eqref{app8.-ncp(phi)-integral_form} are monotonically increasing and decreasing functions of $\phi$ in the range of interest. Since $\sin(\frac{\xi}{2}){\geq}|\tilde r_{+,x}|$ one obtains the upper estimate $\ncp_{\mathrm{I}}{\leq}\ncp_{\mathrm{I,max}}$, where:
\begin{gather}\label{app8.-n(r_-,r_+)}
\ncp_{\mathrm{I,max}}{=}\ncp|_{\phi{=}\frac1{|\tilde r_{+,x}|}}{=}\frac{\arccos(\frac{\tilde r_{-,x}}{\tilde r_{+,x}})}{2\arctan(\frac{\umax}{|\tilde r_{+,x}|})}{+}1.
\end{gather}
In order to make this result constructive we need estimates for $r_{\pm,x}$.
Elementary analysis shows that $\ncp_{\mathrm{I,max}}(\tilde r_{+,x},\tilde r_{-,x})$ is a monotonic function of $\tilde r_{-,x}$ and reaches a maximum when $\Sign(\tilde r_{+,x})\tilde r_{-,x}$ is minimal. At the same time,
$\ncp_{\mathrm{I,max}}(\tilde r_{+,x},\tilde r_{-,x})$ is a concave function of $\tilde r_{+,x}$ when $\tilde r_{+,x}\tilde r_{-,x}{<}0$ and monotonically increasing function of $|\tilde r_{+,x}|$ in the range $\tilde r_{+,x}\tilde r_{-,x}{>}0$. Thus, the upper estimate for $\ncp_{\mathrm{I,max}}$ can be calculated by substituting into \eqref{app8.-n(r_-,r_+)} appropriate upper and/or lower estimates for $\tilde r_{+,x}$ and $\tilde r_{-,x}$.
Specifically, according to Proposition~\ref{*proposition-r_i-global_optimality_bounds} $|\tilde r_{-,x}|{<}|\tilde r_{+,x}|{<}|\rP_x|$, and $0{<}|\tilde r_{-,x}|{<}|\rM_x|$. Substitution of these estimates results in \eqref{*proposition-n_max_for_time-optimal_case_refined(mp<0)} for the case $\rP_x\rM_x{<}0$ and the second of the estimates \eqref{*proposition-n_max_for_time-optimal_case_refined(mp>0)} for the case $\rP_x\rM_x{>}0$.

Note that the latter estimate directly accounts for the location of only one trajectory endpoint and can be further refined. Namely, due to \eqref{*proposition-r_i-global_optimality_bounds_inequality} the corner points in the case $\tilde r_{+,x}\tilde r_{-,x}{>}0$ are located in the range $\tilde r_{i,x}{\in}[0,\rP_x]$.
Since the $x$-coordinates of the corner points  are monotonic functions of the index $i$ (see Proposition~\ref{*proposition-n_max_for_time-optimal_case} and Fig.~\ref{@FIG.05}), the trajectory can be split into two continuous parts $R_1$ and $R_2$ such that all $\ncp_{R_1} (\ncp_{R_2})$ corner points in the segment $R_1 (R_2)$ belong to the range $\tilde r_{i,x}{\in}(\rM_x,\rP_x]$ ($\tilde r_{i,x}{=}[0,\rM_x]$), and their junction point $\tilde r_c$ is chosen such that $\tilde r_{c,x}{=}\rM_x$. Using these range estimates and the extremal properties of function~\eqref{app8.-n(r_-,r_+)} we obtain that $\ncp_{R_1}{\leq}\frac{\arccos(\frac{\rM_x}{\rP_x})}{|2\arctan(\frac{\umax}{\rP_x})|}{+}1$. Let us show that $\ncp_{R_2}{\leq}3$ (which will prove the first estimate in \eqref{*proposition-n_max_for_time-optimal_case_refined(mp>0)}). Indeed, the duration $\tilde{\Delta\tau}_{R_2}$ of this segment can not exceed $\pi$ (the maximal duration of the trajectory with $\tilde u(\tau){=}0$ connecting $\rM$ and $\tilde r_c$). At the same time, according to eq.~\eqref{03.-regular_segment_length} the minimal duration of each arc of the bang-bang trajectory is $\frac{\pi}{2}\cos\alpha$. Thus, the number of the interior bang segments of duration $\tilde{\Delta\tau}$ in the case $\umax{\leq}1$ can not exceed $[2\sqrt{2}]{=}2$, i.e. $\ncp_{R_2}{\leq}3$ (the same restriction for the case $\umax{>}1$ trivially follows from Proposition \eqref{*proposition:n<=pi/alpha}). Hence, Proposition is completely proven.

\end{document}